\documentclass{article}

\usepackage{times}
\usepackage[margin=1in]{geometry}
\usepackage[utf8]{inputenc}
\usepackage[T1]{fontenc}
\usepackage{hyperref}
\usepackage{url}
\usepackage{booktabs}
\usepackage{amsfonts}
\usepackage{amsmath}
\usepackage{amssymb}
\usepackage{amsthm}
\usepackage{nicefrac}
\usepackage{microtype}
\usepackage{graphicx}
\usepackage[table]{xcolor}
\usepackage{subcaption}
\usepackage{makecell}
\usepackage{multirow}
\usepackage{xspace}
\usepackage{natbib}
\usepackage{algorithm}
\usepackage{algorithmic}
\usepackage{enumitem}
\usepackage[capitalise]{cleveref}

\newtheorem{theorem}{Theorem}
\newtheorem{proposition}[theorem]{Proposition}
\newtheorem{lemma}[theorem]{Lemma}

\newtheorem{remark}[theorem]{Remark}

\newcommand{\bandmf}{\texttt{DP-BandMF}\xspace}
\newcommand{\dpmf}{\texttt{DP-MF}\xspace}
\newcommand{\dpsgd}{\texttt{DP-SGD}\xspace}
\newcommand{\poisson}{\texttt{Poisson}\xspace}
\newcommand{\ra}{\texttt{Random Allocation}\xspace}
\newcommand{\pld}{\texttt{PLD}\xspace}
\newcommand{\mc}{\texttt{MC}\xspace}

\newcommand{\deterministic}{\texttt{Deterministic}\xspace}
\newcommand{\cyclicpoisson}{\texttt{Cyclic Poisson}\xspace}
\newcommand{\ballsinbins}{\texttt{Balls-in-Bins}\xspace}
\newcommand{\minseppoisson}{\texttt{Min Sep Poisson}\xspace}

\newcommand{\cyclicra}{\texttt{Cyclic RA}\xspace}
\newcommand{\nestedra}{\texttt{Nested RA}\xspace}
\newcommand{\minsepra}{\texttt{Min Sep RA}\xspace}

\title{Better Batches, Better Privacy:\\Random Allocation for DP-BandMF}

\author{
  Ryan McKenna \\
  Google Research \\
  \texttt{mckennar@google.com}
  \and
  Arun Ganesh \\
  Google Research \\
  \texttt{arunganesh@google.com}
}

\newcommand{\boldzero}{\ensuremath{\boldsymbol{0}}}

\newcommand{\bfA}{\ensuremath{\mathbf{A}}}

\newcommand{\bfC}{\ensuremath{\mathbf{C}}}

\newcommand{\bfI}{\ensuremath{\mathbf{I}}}

\newcommand{\bfc}{\ensuremath{\mathbf{c}}}

\newcommand{\bfx}{\ensuremath{\mathbf{x}}}
\newcommand{\bfy}{\ensuremath{\mathbf{y}}}
\newcommand{\bfz}{\ensuremath{\mathbf{z}}}

\newcommand{\calD}{\ensuremath{\mathcal{D}}}

\newcommand{\calM}{\ensuremath{\mathcal{M}}}
\newcommand{\calN}{\ensuremath{\mathcal{N}}}

\newcommand{\calS}{\ensuremath{\mathcal{S}}}

\newcommand{\ltwo}[1]{\left\|#1\right\|_2}

\begin{document}

\maketitle

\begin{abstract}
The current state of the art for batch selection and noise addition in differentially private training is a combination of DP Banded Matrix Factorization (DP-BandMF; informally, DP-SGD where the noise added in two different steps can be correlated if the steps are within $b$ of each other) and privacy amplification by sampling via Monte Carlo accounting. First, whereas prior batch selection strategies for DP-BandMF were inspired by Poisson sampling, random allocation is known to outperform Poisson sampling for DP-SGD. We therefore propose and analyze new batch selection strategies for DP-BandMF inspired by random allocation, improving RMSE by up to $1.85\times$ over previous work. Second, Monte Carlo accounting typically requires drawing $\Omega(1/\delta)$ samples from a privacy loss distribution, making accounting costly. We increase sampling and scoring throughput by two orders of magnitude through a carefully optimized GPU implementation. Finally, previous work chose the DP-BandMF strategy matrix and the batch selection strategy independently. We give a simple framework for choosing these parameters jointly, reducing RMSE by up to $13.4\%$ over fixed strategy matrices.
\end{abstract}

\section{Introduction}

Differential privacy (DP) \citep{dwork2006calibrating} guarantees that an algorithm's output on two datasets $D, D'$ differing in one example satisfy some indistinguishability guarantee. The de facto DP ML algorithm is differentially private stochastic gradient descent (\dpsgd) \citep{abadi2016deep}, which clips each per-example gradient to a maximum norm, and adds Gaussian noise to their sum.

A key technique for \dpsgd's success is \textit{privacy amplification by sampling} \citep{kasiviswanathan2008what}. Privacy amplification by sampling is the observation that by randomizing batches in \dpsgd, the problem of guessing whether $D$ or $D'$ was used becomes strictly harder, allowing us to reduce the noise while retaining the same privacy guarantee. For most of the time since the popularization of \dpsgd by the seminal work of \citep{abadi2016deep}, the default method for randomizing batches (which we refer to as a \textit{batch selection strategy}) has been \textit{Poisson sampling}: each example in the dataset is included in each batch independently with probability $p$. 

Recent implementations of data processing pipelines \citep{chua2024scalable, grain2023github} allow for random access to data elements and arbitrary batch sizes, making almost any batch selection strategy readily implementable. In turn, choosing batch selection strategies which admit lower noise multipliers is a low-barrier, high-yield path to improving practical DP training. In light of this, several recent works have proposed alternatives to Poisson sampling, improving on it in two directions:

\begin{itemize}
    \item \textbf{Controlling participation variance \citep{shenfeld2026privacy, dong2026balanced, dong2026optimal}}: A downside of Poisson sampling is that an example can participate many times, meaning in some extreme settings Poisson sampling is actually worse than deterministic batching for privacy \citep{chua24howprivate}. A natural correction is \textit{random allocation\footnote{As named by \cite{shenfeld2026privacy}; it is also called balanced iteration subsampling by the concurrent work of \cite{dong2026balanced}.}:} fix the number of participations, and only randomize their position.
    \item \textbf{Matrix factorization \citep{choquettechoo2025nearexact, dong2025bminsep}}: \dpsgd adds independent noise in each step. \dpmf is a generalization that allows one to add noise that is correlated in different steps. DP Matrix Factorization (\dpmf) breaks a standard assumption in privacy analysis that each step's randomness is independent. \cite{choquettechoo2025nearexact}  bypassed this assumption by analyzing privacy amplification by sampling for \dpmf via Monte Carlo accounting, where a sufficient primitive is to compute the likelihood ratio for an output of the algorithm on adjacent datasets. For Poisson sampling this primitive is inefficient, but several works have proposed alternative batch selection strategies for which this primitive is efficient. 
\end{itemize}

While for \dpsgd \cite{dong2026optimal} provides evidence that random allocation is near-optimal for \dpsgd, for \dpmf there is still a rich yet largely unexplored space of batch selection strategies. As \dpmf can substantially improve utility in practice compared to \dpsgd \citep{xu2023federated}, identifying the best combination of privacy amplification and \dpmf is an important practical question. In this paper we attempt to enumerate and compare some of these strategies.

\paragraph{Contributions}
We connect several ideas presented in recent prior work: (1) random allocation for \dpsgd \citep{shenfeld2026privacy, dong2026balanced}, (2) Monte Carlo accounting for \dpmf \citep{choquettechoo2023amplified, dong2025bminsep}, and (3) parametric \dpmf optimization under Monte Carlo accounting \citep{kalinin2026dp, mckennaScaling2024}.

\begin{itemize}
    \item \textbf{Random Allocation for \bandmf:} We first attempt to identify the best batch selection strategy for DP Banded Matrix Factorization (\bandmf), a subclass of \dpmf that is particularly amenable to privacy amplification. We introduce three new randomized batch selection strategies for \bandmf that guarantee exact per-example participation counts while enforcing a useful ``$b$-separation'' property in \cref{sec:new_strategies}. We derive exact recurrence relations for their likelihood ratios to enable Monte Carlo accounting in \cref{sec:recurrences}. Across all evaluated settings our proposed strategies consistently improve utility, achieving a median $1.54\times$ (up to $1.85\times$) improvement in RMSE (the standard utility metric for \dpmf) over the baseline, and a median $1.10\times$ (up to $1.27\times$) improvement over the prior state of the art.
    
    \item \textbf{Efficient Monte Carlo Accounting:} We next aim to make Monte Carlo accounting more efficient. We show that the structure of our batch selection strategies and Toeplitz matrices make Monte Carlo accounting particularly well-suited to GPU acceleration in \cref{sec:implementation}. Through a carefully-optimized JAX implementation, we demonstrate throughput exceeding $10^6$ privacy loss samples per second on an H100 GPU in \cref{fig:throughput}, representing a $\approx50\times$ to $\approx100\times$ speedup over existing CPU baselines and making the technique more practical and accessible.
    
    \item \textbf{Strategy Optimization under Monte Carlo Accounting:} The above contributions assume the \bandmf strategy matrix is fixed in advance. We propose a simple and efficient method to optimize over a single-parameter class of strategy matrices under Monte Carlo accounting, and demonstrate that this can reduce RMSE over fixed strategies by up to $13.4\%$.
\end{itemize}

\section{Background}

\subsection{Differential Privacy, \dpsgd, and \dpmf}

Differential privacy \citep{dwork2006calibrating} is the gold standard for data privacy. Given a data universe $\calD$, we say two datasets $D, D' \in \calD^*$ are adjacent, denoted $D \simeq D'$, if $D'$ is the same as $D$ except with one example removed or vice-versa\footnote{This is the \textit{add-remove adjacency}. All results in the paper also hold for the \textit{zero-out adjacency}.}. Given two distributions $P, Q$ over an output space $\mathcal{Y}$, the $\alpha$-hockey stick divergence is $   H_\alpha(P, Q) := \mathbb{E}_{y \sim P}\!\left[\max\!\left\{1 - \alpha \frac{Q(y)}{P(y)}, 0\right\}\right].$ A mechanism $\calM: \calD^* \rightarrow \mathcal{Y}$ is $(\varepsilon, \delta)$-DP if for all $D \simeq D'$, we have $H_{e^\varepsilon}(\calM(D), \calM(D')) \leq \delta$.

\dpsgd \citep{abadi2016deep} is the standard approach for DP machine learning. In iteration $i \in [n]$, letting
$\textsf{clip}({\theta}, {C}) := \theta \cdot \frac{C}{\max\{\ltwo{\theta}, C\}},$ given batch $B_i$, clipping norm $C$, and noise multiplier $\sigma$, we compute the noisy gradient $\tilde{g}_i = \sum_{d \in B_i} \textsf{clip}(\nabla \ell(\theta_{i-1}; d), C) + \bfz_i,$ where $\bfz_i \sim \calN(\boldzero, \sigma^2 C^2 \bfI)$. The noisy gradient can then be used to perform a model update (from $\theta_{i-1}$ to $\theta_i$) by any first-order optimizer. Without loss of generality, we assume $C = 1$ throughout the paper.

\dpmf \citep{denisov2022improved} is a generalization of \dpsgd which uses noise correlated across steps. For simplicity we present it in the one-dimensional optimization setting. Let $\bfz \sim \calN(\boldzero, \sigma^2 \bfI_T)$. \dpmf is parameterized by a (lower-triangular) matrix $\bfC \in \mathbb{R}^{T \times T}$. \dpmf adds noise $(\bfC^{-1} \bfz)_i$ in iteration $i$, whereas \dpsgd would add noise $\bfz_i$. The setting $\bfC = \bfI$ retrieves \dpsgd. \dpmf has been observed to achieve substantial improvements in utility over \dpsgd \cite{denisov2022improved} including in production training runs \cite{xu2023federated}.

The main variant of \dpmf we are interested in is ($b$-)\bandmf \citep{choquettechoo2023amplified}, in which $\bfC$ is $b$-banded, i.e. $\bfC[i, j] = 0$ when $i \not \in [j, j + b - 1]$. A line of works \citep{choquettechoo2023amplified, choquettechoo2025nearexact, dong2025bminsep} have observed \bandmf is particularly amenable to privacy amplification by sampling. For convenience we also throughout the paper assume Toeplitz $\bfC$, i.e. those with constant diagonals, equivalently $\bfC_{i,j} = \bfC_{i+k, j+k}$ for all $i, j, k$.

\subsection{Privacy Analysis and Monte Carlo Accounting}

Privacy amplification by sampling \citep{kasiviswanathan2008what} is the observation that randomizing the batches used in \dpsgd and \dpmf can improve the privacy-utility tradeoff ``for free,'' i.e. without modifying training parameters like the number of iterations or (average) batch size. Hence, the problem of analyzing the privacy guarantees of \dpsgd or \dpmf with randomized batch selection strategies is fundamental for optimizing utility in practice. Lemma 4.5 of \cite{choquettechoo2024privacyamplification}, restated below for convenience, simplifies this task:

\begin{lemma}\label{lem:dimensionreduction}
Suppose we run $T$-step \dpmf using a batch selection strategy such that each example participates independently from each other. For a given example, let $\bfx \in \mathbb{R}^T$ be the random vector where $\bfx_i = 1$ if that example is included in batch $B_i$ and $0$ otherwise, and we assume $\bfx$ is distributed identically for all examples. Then \dpmf using this batch selection strategy is $(\varepsilon, \delta)$-DP iff $H_{e^\varepsilon}(P, Q) \leq \delta$ and $H_{e^\varepsilon}(Q, P) \leq \delta$, where $P = \bfC \bfx + \bfz$ and $Q = \bfz$, $\bfz \sim \calN(\boldzero, \sigma^2 \bfI_T)$.    
\end{lemma}

Analyzing $H_{e^\varepsilon}(P, Q)$ is challenging because $P$, $Q$ are high-dimensional distributions and $P$ does not necessarily have a simple structure. This is in contrast to e.g. \dpsgd with Poisson sampling, where the coordinates of $P$ are independent. This enables techniques like composition theorems which reduce to easy-to-analyze one-dimensional distributions. There are two predominant techniques in the literature for handling privacy amplification of \dpmf:
\begin{itemize}
    \item \cite{choquettechoo2023amplified} propose ``cyclic'' batch selection strategies for \bandmf. Each cyclic strategy splits the dataset (arbitrarily) into $b$ subsets. The $i$th subset is only eligible to participate in steps $j$ where $j \equiv i \pmod b$. Within these steps, they use another ``inner'' batch selection strategy such as Poisson sampling to subsample from the eligible examples. They prove the privacy guarantees reduce to those of $T/b$-step \dpsgd using the inner batch selection strategy. 
    \item Monte Carlo accounting \citep{wang2023randomized} uses Monte Carlo estimation of $H_{e^\varepsilon}(P, Q)$. In order to perform Monte Carlo accounting, it suffices to have primitives for (i) sampling $y \sim P$, (ii) evaluating $P(y)/Q(y)$ for any given $y$. (i) is straightforward since sampling $\bfx$ is equivalent to implementing the batch selection strategy, but (ii) can be more challenging since $\bfx$ can take on up to $2^T$ values. \cite{dong2025bminsep} gave a procedure which uses Monte Carlo accounting, and allows one to run end-to-end accounting and training with formal $(\varepsilon, \delta)$-DP, including calibration of the $\sigma$ parameter. We restate their procedure and its guarantees in \cref{sec:mc-background} for convenience. Notably, Monte Carlo accounting typically requires $> 1/\delta$ samples, making it much more expensive than standard accounting. While this adds significant overhead, \cite{dong2025bminsep} demonstrated that Monte Carlo accounting is still feasible at production scales.
\end{itemize}

\subsection{Prior Batch Selection Strategies}
\label{sec:prior_strategies}

We briefly review the batch selection strategies from prior work that are most relevant to our work:

\paragraph{Poisson sampling (\dpsgd):} Each example is included in each batch independently with probability $p$, yielding an expected participation count of $\mathbb{E}[N] = pT$ \citep{abadi2016deep}. While analytically convenient, an example's total participation count follows a Binomial distribution whose heavy right tail can degrade privacy \citep{chua24howprivate}.

\paragraph{Random Allocation (\dpsgd):} Each example independently picks a uniformly random subset of $K$ out of the $T$ batches to be included in. This batch selection strategy was proposed in several concurrent works \citep{shenfeld2026privacy, dong2026balanced}. \citet{dong2026optimal} gave a principled theoretical explanation and strong empirical evidence for why \dpsgd with random allocation seems to strictly improve on \dpsgd with Poisson sampling.

\paragraph{Deterministic (\bandmf):} Each example participates at fixed, evenly spaced iterations with gap $ T/K \ge b$ (we call the property of having size $\geq b$ gaps \textit{$b$-separation}). Because noise in \bandmf is correlated across windows of length $b$, $b$-separation ensures that an example's participations do not overlap in correlated noise, enabling a simpler privacy analysis. 

\paragraph{Balls-In-Bins (\bandmf):} Each example participates at evenly spaced iterations again, but the step of the first participation is randomized. \cite{choquettechoo2025nearexact} showed this is amenable to Monte Carlo accounting as $\bfx$ only takes on a few values.

\paragraph{Cyclic Poisson (\bandmf):} Each example is assigned to one of $b$  subsets $i \in [b]$ arbitrarily and participates independently with probability $p$ only in iterations $t \equiv i \pmod b$ \citep{choquettechoo2023amplified}. This provides $b$-separation, but still suffers from participation variance.

\paragraph{Min Sep Poisson (\bandmf):} The same as Poisson sampling, but rejecting any example that participated in the previous $b-1$ iterations \citep{dong2025bminsep}. This maintains $b$-separation while achieving better utility than cyclic Poisson,  but again suffers from participation variance.

As summarized in \cref{tab:batch_selection_taxonomy}, each existing strategy satisfies at most two of three desirable properties: fixed participation counts, $b$-separation, and maximum benefit from amplification. We formalize these design principles in \cref{sec:desiderata}.

\definecolor{cgreen}{RGB}{38, 138, 64}
\definecolor{xred}{RGB}{204, 51, 51}

\begin{table}[t]
\centering
\small
\caption{Taxonomy of batch selection strategies.}
\label{tab:batch_selection_taxonomy}
\begin{tabular}{lcccccc}
\toprule
\textbf{Strategy} & \textbf{Mechanism} & \textbf{Params} & \textbf{Fixed $K$} & \textbf{$b$-Sep} & \textbf{Entropy} & \textbf{Accounting} \\
\midrule
\poisson & \dpsgd & $p$ & \textcolor{xred}{$\times$} & \textcolor{xred}{$\times$} & \textcolor{cgreen}{\checkmark\checkmark} & \pld \\
\ra & \dpsgd & $K$ & \textcolor{cgreen}{\checkmark} & \textcolor{xred}{$\times$} & \textcolor{cgreen}{\checkmark\checkmark} & \pld \\
\deterministic & \dpmf + \bandmf & $K$ & \textcolor{cgreen}{\checkmark} & \textcolor{cgreen}{\checkmark} & \textcolor{xred}{$\times$} & \pld \\
\cyclicpoisson & \bandmf & $b, p$ & \textcolor{xred}{$\times$} & \textcolor{cgreen}{\checkmark} & \textcolor{cgreen}{\checkmark} & \pld \\
\ballsinbins & \dpmf + \bandmf & $b, p$ & \textcolor{cgreen}{\checkmark} & \textcolor{cgreen}{\checkmark} & \textcolor{cgreen}{\checkmark} & \pld \\
\minseppoisson & \bandmf & $b, p$ & \textcolor{xred}{$\times$} & \textcolor{cgreen}{\checkmark} & \textcolor{cgreen}{\checkmark\checkmark} & \mc \\
\midrule
\cyclicra & \bandmf & $b, K$ & \textcolor{cgreen}{\checkmark} & \textcolor{cgreen}{\checkmark} & \textcolor{cgreen}{\checkmark} & \pld \\
\nestedra & \bandmf & $b, K$ & \textcolor{cgreen}{\checkmark} & \textcolor{cgreen}{\checkmark} & \textcolor{cgreen}{\checkmark\checkmark} & \mc \\
\minsepra & \bandmf & $b, K$ & \textcolor{cgreen}{\checkmark} & \textcolor{cgreen}{\checkmark} & \textcolor{cgreen}{\checkmark\checkmark} & \mc \\
\bottomrule
\end{tabular}
\end{table}

\subsection{Related Work} \label{sec:related}

Our work operates at the intersection of privacy accounting, batch selection, and noise addition, which are three key building blocks for \dpsgd-style mechanisms (along with gradient clipping).

\paragraph{Privacy accounting.}
Tight privacy accounting provides utility gains without altering the training dynamics.
Research in this area has progressed from RDP-based analyses \citep{abadi2016deep,mironov2017renyi,balle2018privacy,feldman2018privacy,balle2020hypothesis} 
to tight numerical composition via PLD \citep{koskela2020computing,gopi2021numerical}.
As batching and noise mechanisms have grown more complex, accounting has shifted toward tight mechanism-specific analyses \citep{shenfeld2026privacy,feldman2026efficientprivacylossaccounting,schuchardt2026samplingfree,dong2025bminsep} and Monte Carlo accounting \citep{wang2023randomized,choquettechoo2025nearexact}.

\paragraph{Batch selection.}
Privacy analyses have traditionally assumed independent \poisson subsampling \citep{abadi2016deep,balle2018privacy} to simplify the math, while in practice training pipelines typically rely on fixed-size batches and sequential data access, creating a well-documented mismatch between theoretical accounting assumptions and actual implementations \citep{rodriguez2024efficient,chua2024scalable,chua24howprivate}.
Beyond the selected batch selection strategies highlighted in the main paper, other ones that have been proposed in practice include shuffling, sampling with replacement, and sampling without replacement \cite{balle2018privacy}. For cross-device federated learning settings where random access to individual users or records is infeasible, minimum separation (without any subsampling) is another batch selection strategy that has been proposed and analyzed in prior work \citep{kairouz2021practical,choquettechoo2023amplified,ganesh2025s}.

\paragraph{Noise addition.}
Rather than injecting independent Gaussian noise at each step as in standard \dpsgd, \dpmf adds noise that is correlated across iterations to reduce the error on the \emph{prefix sums} of gradients at the cost of increased error on individual gradients \citep{choquette2024correlated,denisov2022improved,choquette2022multi}. Modern approaches to this are based on the \emph{matrix mechanism}, originally developed for answering linear query workloads with DP \citep{li2015matrix,mckenna2018optimizing}. Work in this space includes advancing SOTA in RMSE \citep{choquettechoo2023amplified,choquettechoo2024privacyamplification,dong2025bminsep}, improving memory efficiency without sacrificing too much utility  \citep{mckennaScaling2024,dvijotham2024efficient,kalinin2026dp}, optimizing strategy matrices for realistic training scenarios rather than proxies based on RMSE \citep{koloskova2023gradient,gu2025correlating}, and theoretical analyses \citep{fichtenberger2023constant,henzinger2023almost}.

\paragraph{Gradient clipping.}
An orthogonal line of work improves per-example gradient clipping, a key primitive for \dpsgd and \dpmf.
Key innovations include adaptive and quantile-based clipping thresholds \citep{pichapati2019adaclip,andrew2021differentially}, computational methods such as ghost clipping and per-layer clipping \citep{bu2022scalable,bu2023automatic,he2022exploring}, and geometry-aware clipping strategies \citep{li2022differentially,ganesh2025design,gilani2026geoclip}.

\paragraph{\bandmf variants.}
This work builds directly on the following foundational advances in correlated-noise mechanisms, batch selection, and Monte Carlo privacy accounting:
\begin{itemize}
    \item \citet{choquettechoo2023amplified} introduced the $b$-banded noise correlation family (\bandmf) and cyclic reduction to \dpsgd, establishing the correlated-noise mechanism and cyclic baseline (\cyclicpoisson) that we build on.
    \item \citet{choquettechoo2025nearexact} laid the foundation for Monte Carlo privacy accounting and strategy optimization in \bandmf via the \ballsinbins batch selection strategy.
    \item \citet{dong2025bminsep} developed the dynamic-programming framework for Monte Carlo likelihood-ratio scoring and applied it to the \minseppoisson batch selection strategy.
    \item \citet{shenfeld2026privacy,dong2026optimal} demonstrated that fixing per-example participation counts (random allocation) consistently improves privacy-utility trade-offs over Poisson sampling in \dpsgd, providing the batch selection strategy that inspired this work.
    \item \citet{mckennaScaling2024} showed that simple single-parameter Toeplitz families (such as $1/i + r$) achieve near-optimal prefix-sum RMSE, enabling our gradient-free joint $(b, r)$ search under Monte Carlo accounting.
\end{itemize}

\section{Batch Selection Strategies for \bandmf}

In this section, we propose and analyze a family of batch selection strategies for \bandmf based on random allocation. Below, we formalize the design principles, present our proposed batch selection strategies, derive their associated recurrence relations, and provide practical implementation details. For simplicity, we assume the number of bands $b$ evenly divides the number of iterations $T$ throughout the paper. We also assume $\bfC$ is $b$-banded Toeplitz, and the first column of $\bfC$ has norm $1$.

\subsection{Design Principles} \label{sec:desiderata}

Throughout this work we attempt to answer several fundamental questions about the combination of \bandmf and privacy amplification. We start with arguably the most immediate one: \textit{What is the batch selection strategy which maximally reduces the noise used by \bandmf?} Below we enumerate a set of desiderata that reasonable batch selection strategies should satisfy.

\begin{itemize}
    \item \textbf{Fixed-K}: Each example should participate a fixed number of times $K$. As \cite{dong2026optimal} demonstrate, this is the ``first-order'' predictor of the privacy-utility tradeoff.
    \item \textbf{$b$-Separation:} Each participation from the same example should be separated by $\geq b$ steps. As observed in \cite{choquettechoo2023amplified, choquettechoo2025nearexact}, this simplifies privacy analysis for \bandmf. 
    \item \textbf{Entropy}: All else equal, assignments of examples to iterations that are ``more random'' are generally preferred. As \cite{dong2026optimal} observed, more randomization helps, but is a ``second-order'' predictor of the privacy utility tradeoff compared to fixed-$K$.
    \item \textbf{Structure:} The following properties greatly simplify privacy analysis: (1) The batch assignment should be independent across examples (as required by \cref{lem:dimensionreduction}), and (2) the sampling should have a simple recursive structure to make Monte Carlo accounting tractable.
\end{itemize}

\subsection{Proposed Strategies} \label{sec:new_strategies}

To address the desiderata outlined in \cref{sec:desiderata}, we propose three new batch selection strategies for \bandmf that generalize random allocation \citep{shenfeld2026privacy, dong2026balanced}. Below, we formalize each strategy from the perspective of a single example's participation schedule.

\paragraph{Cyclic Random Allocation (\cyclicra).}
Our first batch selection strategy, \textit{cyclic random allocation} (\cyclicra), is an instantiation of the cyclic framework of \citet{choquettechoo2023amplified}, replacing the inner Poisson sampling with random allocation. Specifically, the dataset is partitioned into $b$ subsets (arbitrarily, but consistently across neighboring datasets $D \simeq D'$), and each example in subset $i \in [b]$ is eligible to participate only in rounds $t \equiv i \pmod b$. Within these eligible rounds, each example selects exactly $K$ rounds uniformly at random. Just as random allocation improves upon Poisson sampling for \dpsgd, we expect \cyclicra to improve upon \cyclicpoisson. As we show in \cref{sec:recurrences}, \cyclicra inherits the reduction of \citet{choquettechoo2023amplified}, allowing near-exact privacy analysis via standard privacy loss distributions (PLD).

\begin{figure}[t] 
\begin{minipage}[t]{0.42\textwidth}
\begin{algorithm}[H]
\caption{\nestedra}
\label{alg:nested_ra_def}
\begin{algorithmic}[1]
\REQUIRE Parameters $(T, K, b)$
\STATE Sample $i \sim \text{Uniform}([b])$
\STATE $\mathcal{E} \leftarrow \{t \in [T] : t \equiv i \pmod b\}$
\STATE Randomly sample $S \subseteq \mathcal{E}$ of size $K$
\RETURN $S$
\end{algorithmic}
\end{algorithm}
\end{minipage}
\hfill
\begin{minipage}[t]{0.58\textwidth}
\begin{algorithm}[H]
\caption{\minsepra}
\label{alg:balanced_min_sep_def}
\begin{algorithmic}[1]
\REQUIRE Parameters $(T, K, b)$ with $T \ge Kb$
\STATE Let $\mathcal{V} \leftarrow \big\{S \subset [T] : |S| = K, | s_i - s_j | \geq b \forall i \neq j \} $
\STATE Sample $S \sim \text{Uniform}(\mathcal{V})$
\RETURN $S$
\end{algorithmic}
\end{algorithm}
\end{minipage}
\caption{Conceptual definitions of our batch selection strategies applied to a single example. \nestedra restricts candidates to a residue class modulo $b$, then samples $K$ steps uniformly without replacement. \minsepra samples uniformly at random from the set of all $K$-element subsets that satisfy the minimum separation $b$. \cyclicra is the same as \nestedra, but the residue class $i$ can be chosen arbitrarily rather than randomly.}
\label{fig:sampling_definitions}
\end{figure}

\paragraph{Nested Random Allocation (\nestedra).}
A limitation of \cyclicpoisson and \cyclicra is that the choice of modulus $i \in [b]$ for each example is arbitrary, when it could be randomized for better privacy. \textit{Nested random allocation} (\nestedra) randomizes the modulus by assigning each example uniformly at random to one of the $b$ subsets. It is otherwise identical to \cyclicra. A formal specification of \nestedra is given in \cref{alg:nested_ra_def} , written as a sampler for a set of steps for a single example. Because participations under different moduli can affect the same output coordinates, the reduction of \citet{choquettechoo2023amplified} cannot account for this added randomness; we instead analyze \nestedra via Monte Carlo accounting in \cref{sec:recurrences}.

\paragraph{Min Sep Random Allocation (\minsepra).}

Just as \minseppoisson improves upon \cyclicpoisson by enforcing $b$-separation only on \textit{actual} participations rather than \textit{potential} participations, we do the same but for random allocation. In \textit{min sep random allocation} (\minsepra), each example independently selects a uniformly random subset of $K$ rounds from $[T]$ satisfying the $b$-separation property ($|s_i - s_j| \ge b$ for all $i \neq j$). A formalization of \minsepra is given in \cref{alg:balanced_min_sep_def}.
The following proposition enables an efficient implementation:

\begin{proposition}[Schedule Bijection]\label{prop:bijection}
Let $T' = T - (K-1)(b-1)$. The map $\phi(s_1, \dots, s_K) = (s_1, s_2 + (b-1), \dots, s_K + (K-1)(b-1))$ for sorted $s_1 < \dots < s_K \in [T']$ is a bijection from $\binom{[T']}{K}$ to the set of valid $b$-separated subsets of $[T]$ of size $K$.
\end{proposition}

We analyze the privacy of \minsepra via Monte Carlo accounting in \cref{sec:recurrences}. While \minsepra has higher entropy than e.g. \nestedra, it has one notable drawback:

\begin{remark}[Expected Batch Sizes]
\label{rem:batch_size_stability}
For \minsepra, the expected batch size varies across steps: near the beginning and end of training, the expected batch size is larger. Similar to the warm-start correction used for \minseppoisson \citep{dong2025bminsep}, there is a natural correction to the batch selection strategy which fixes this behavior. See \cref{sec:biased_batches} for more details.
\end{remark}

\subsection{Privacy Analyses}
\label{sec:recurrences}

We now provide formal statements for privacy analysis of our proposed batch selection strategies.

\paragraph{Cyclic Random Allocation} 
Theorem 5 of \cite{choquettechoo2023amplified} gives the privacy analysis:

\begin{theorem}
$b$-\bandmf using $K$-out-of-$T$ cyclic random allocation satisfies any privacy guarantee satisfied by \dpsgd using $K$-out-of-$T/b$ random allocation.
\end{theorem}

The privacy guarantees for \dpsgd with random allocation can be computed using the results of \cite{feldman2026efficientprivacylossaccounting, dong2026optimal}.

\paragraph{Nested Random Allocation} 

The analysis of \cite{choquettechoo2023amplified} assumes the modulus is public and doesn't benefit directly from randomization of it. Instead of using their analysis, we provide the necessary primitives for Monte Carlo accounting: (i) being able to sample $\bfC \bfx + \bfz$ efficiently, where $\bfx, \bfz$ are defined as in \cref{lem:dimensionreduction}, and (ii) being able to evaluate $P(\bfy)/Q(\bfy)$ for an arbitrary vector $\bfy$, where $P = \bfC \bfx + \bfz$ and $Q = \bfz$. (i) is straightforward, so we focus on (ii).

We cannot evaluate $P/Q$ via naive summation as there are $b \cdot \binom{T/b}{K}$ possible values for $\bfx$. We instead follow the framework of \cite{dong2025bminsep}: we determine a recursive structure that allows us to evaluate $P/Q$ via a dynamic program. Let $\bfc = \bfC[1:b, 1]$, and let $\ell_t := \frac{\calN(\bfc, \sigma^2 \bfI)(\bfy[t:t+b-1])}{\calN(\boldzero, \sigma^2 \bfI)(\bfy[t:t+b-1])}$\footnote{We abuse notation here for readability; if $t > T - b + 1$, instead we are looking at the last $T - t + 1 < b$ coordinates of $\bfy$, and the first $T - t + 1$ coordinates of $\bfc$.}. That is, $\ell_t$ is the likelihood ratio of $\bfy[t:t+b-1]$, assuming participation in step $t$. Let $\textsc{LR}(k, t)$ denote the likelihood ratio of the length $T-t+1$ suffix $\bfy[t:T]$, conditioned on (i) the example not having participated in iterations $t-b+1$ to $t-1$, (ii) the modulus of the example being $\equiv t \pmod b$, and (iii) the example participates $k$ times in iterations $t$ to $T$. Our base cases are $\textsc{LR}(k, t) = 1$ if either $t > T$ or $ k = 0$. Letting $p_{k,t} := \frac{k}{\lceil(T + 1 - t)/b\rceil}$ our recurrence is:
\begin{equation}
\textsc{LR}(k, t) = p_{k,t} \cdot \ell_t \cdot \textsc{LR}(k-1, t+b) + \left(1 - p_{k,t}\right) \cdot \textsc{LR}(k, t+b).\label{eq:nra-dp}
\end{equation}

\begin{theorem}\label{thm:nra-dp}
    Using Eq. \eqref{eq:nra-dp}, $\frac{1}{b} \sum_{i=1}^b \textsc{LR}(K, i)$ is equal to $P(\bfy)/Q(\bfy)$ for $b$-\bandmf with \nestedra and can be computed in time $O(T(K + \log T))$.\footnote{We remark that a ``screening'' technique was introduced by \cite{dong2026optimal}, that makes the runtime per sample closer to $O(T \log T)$. A similar technique can be used to speed up Monte Carlo accounting using \cref{thm:nra-dp}. Since the changes are analogous to the work of \cite{dong2026optimal}, we omit the details here.}
\end{theorem}

Intuitively, the first (resp. second) term captures the probability of participating (resp. not participating) in step $t$, the likelihood ratio $\ell_t$ (resp. 1) of $\bfy[t:t+b-1]$, and the likelihood ratio of the remaining suffix with $k-1$ (resp. $k$) remaining participations.  

\begin{proof}[Proof of \cref{thm:nra-dp}]

To prove correctness of the recurrence, we first start with the base cases: If $k = 0$ we have $\bfx = 0$ in all remaining coordinates, hence the suffixes of $\bfC \bfx  + \bfz, \bfz$ are distributionally equivalent. For the case $t > T$, note that when the remaining rounds $t$ to $T$ can only accommodate $k$ more participations, we will always have $p_{k,t} = 1$, i.e. the second term in \cref{eq:nra-dp} is $0$. This means that the only way to reach any base cases where $t > T$ and $k > 0$ go through a path through the dynamic programming table where the base case is eventually multiplied by 0. In turn, any base case value for $t > T$ gives correctness.

Now, to prove the recurrence, there are two cases. The first, which happens with probability $\frac{k}{\lceil(T + 1 - t)/b\rceil}$ ($\lceil(T + 1 - t)/b\rceil$ is the number of possible iterations the example can participate in starting from iteration $t$), is that the example participates in round $t$. In this case, for $P$ we have $(\bfC \bfx)[t:t+b-1] = \bfc$, hence the likelihood ratio of coordinates $t$ to $t+b-1$ is exactly $\ell_t$. After iteration $t+b-1$, there are $k-1$ participations remaining, and the likelihood ratio of the subsequent coordinates is independent of all past history once we condition on $k-1$ remaining participations, so $\textsc{LR}(k-1, t+b)$ correctly gives the likelihood ratio for the remaining coordinates. In the second case, where the example doesn't participate in round $t$, we derive the remaining term $\left(1 - \frac{k}{\lceil(T + 1 - t)/b\rceil}\right) \cdot \textsc{LR}(k, t+b)$ analogously.

Our overall likelihood ratio is $\frac{1}{b} \sum_{i=1}^b \textsc{LR}(K, i)$; for each modulus $i$, $\textsc{LR}(K, i)$ gives the likelihood ratio of the entire output conditioned on that modulus (the likelihood ratio of iterations before $i$ is $1$), and each modulus is equally likely.

For the runtime, as we assume $\bfC$ is Toeplitz, each $\ell_t$ can be calculated by taking sliding dot products of $\bfc$ and $\bfy$, which can be done in $O(T \log T)$ time using FFT, and by computing the norm of each window  $\bfc[t:t+b-1]$, which can be done in $O(b)$ time as we assume $\bfC$ is Toeplitz. The recurrence takes $O(1)$ time to compute after pre-computing all these values, and there are $KT$ values to compute the recurrence for.
\end{proof}

\paragraph{Min Sep Random Allocation} 

Due to structural similarities between \minsepra and \nestedra's recursive structure, Monte Carlo accounting for \minsepra is analogous to Monte Carlo accounting for \nestedra. Define $\bfc, \ell_t$ as in the setup for Eq. \eqref{eq:nra-dp}. Let $\textsc{LR}(k, t)$ be the likelihood ratio for $\bfy[t:T]$, conditioned on $k$ participations happening in iterations $t$ to $T$, and the example not participating in iterations $t-b+1, \ldots, t-1$. We let $p_{k,t} := \frac{k}{T - t + 1 - (k-1)(b-1)}.$
We have, similarly to \cref{eq:nra-dp}, the base case $\textsc{LR}(k, t) = 1$ if $k = 0$ or $t > T$ and:
\begin{equation}
\textsc{LR}(k, t) = p_{k,t} \cdot \ell_t \cdot \textsc{LR}(k-1, t+b) + \left(1 - p_{k,t}\right) \cdot \textsc{LR}(k, t+1).\label{eq:msra-dp}
\end{equation}

\begin{theorem}\label{thm:msra-dp}
Using Eq. \eqref{eq:msra-dp}, $LR(K, 1)$ equals $P(\bfy)/Q(\bfy)$ for $b$-\bandmf with \minsepra and can be computed in time $O(T(K + \log T))$.
\end{theorem}

The intuition is analogous to \cref{eq:nra-dp}. Besides the definition of $p_{k,t}$ the only difference is that the second term includes $\textsc{LR}(k, t+1)$ instead of $\textsc{LR}(k, t+b)$, as we no longer restrict participation based on an example's sampled modulus. The first term still uses $\textsc{LR}(k-1, t+b)$ due to enforcement of the $b$-separation constraint.

\begin{proof}[Proof of \cref{thm:msra-dp}]
To prove correctness, the base cases are analogous to \cref{thm:nra-dp}. Conditioned on the example not participating in iterations $t-b+1, \ldots, t-1$, the example participates in $k$ of the remaining $T - t + 1$ iterations while satisfying the $b$-separation property. By the observation that there is a bijection between subsets of $[T]$ of size $K$ satisfying the $b$-separation property, and subsets of $[T-(K-1)(b-1)]$ of size $K$, the probability of the example participating in iteration $t$ is $\frac{k}{T-t+1-(k-1)(b-1)}$. The two terms for when the example does or not does not participate can then be proven as in the proof of \cref{thm:nra-dp}.
\end{proof}

\subsection{Implementation Details} \label{sec:implementation}

\begin{algorithm}[h]
\caption{Monte Carlo Sampling and Scoring \citep{wang2023randomized, dong2025bminsep}}
\label{alg:accounting}
\begin{algorithmic}[1]
\REQUIRE $b$-banded strategy matrix $\bfC \in \mathbb{R}^{T \times T}$, noise multiplier $\sigma$, batch distribution $\mathcal{D}_K$
\STATE Sample sparse participation indicator, $\bfx \sim \mathcal{D}_K, \bfx \in \{0, 1\}^T$ 
\STATE Sample standard Gaussian noise $\bfz \sim \mathcal{N}(\boldzero, \bfI_T)$
\STATE Form observations: $\bfy^+ = \bfC \bfx + \sigma \bfz$ \quad and \quad $\bfy^- = \sigma \bfz$
\STATE Compute scores: $\Lambda^+ \leftarrow \text{ScoreDP}(\bfy^+, \bfC, \sigma, K)$ \quad and \quad $\Lambda^- \leftarrow \text{ScoreDP}(\bfy^-, \bfC, \sigma, K)$
\RETURN $\{\Lambda^+, \Lambda^-\}$
\end{algorithmic}
\end{algorithm}

A high-level algorithm for Monte Carlo Sampling and Scoring for \bandmf is shown in \cref{alg:accounting}. In our actual implementation, we exploit the structure of \bandmf and our proposed batch selection strategies to greatly improve efficiency over a baseline implementation:

\begin{itemize}
\item \textbf{Exploiting $b$-bandedness and $b$-separation} By leveraging the fact that $\bfC$ is a $b$-\emph{banded Toeplitz} matrix and and $\mathcal{D}_K$ samples $b$-separated participations, we can efficiently compute $\bfC\bfx$ in linear time using simple array indexing tricks. Specifically, the $b$-separation guarantees that each participation's contribution is completely disjoint, so $\bfC \bfx$ is simply a zero vector with the $b$ Toeplitz coefficients  $\bfc := \bfC[1:b, 1]$ placed at each index $\{j: \bfx_j = 1\}$ without any collisions.
\item \textbf{Exploiting Fixed-$K$ Structure} Since $|S|=K$ is static under our batch selection strategies, all intermediate variables in \cref{alg:accounting} have static shapes, enabling JIT compilation and efficient GPU-based processing in JAX, which are more efficient than CPU-based numpy implementations.
\item \textbf{Batching and Accumulation} Monte Carlo accounting requires generating $\Omega(\frac{1}{\delta})$ samples for formal verification. To do this efficiently, we leverage \texttt{jax.vmap} and \texttt{jax.lax.scan} to vectorize the computation of many samples and accumulate them into a size $50K$ histogram indexed by privacy loss values linearly spaced out in the range $[-20, 30]$. To ensure the reported privacy guarantee remains formally valid, we round privacy loss values up to the nearest histogram bin value. This combination of higher-order functions allows us to fully saturate the hardware while keeping memory manageable even for very large sample sizes, all while keeping the primary sampling and scoring code specialized to a single sample.
\item \textbf{Carefully Optimized Dynamic Program} Evaluating Eqs. \eqref{eq:nra-dp} and \eqref{eq:msra-dp} requires computing per-step Gaussian log-likelihood ratios $\ell_t$, which in turn requires computing $\langle\bfc, \bfy[t: t+b-1]\rangle$ and $\langle\bfc, \bfc\rangle$, and sweeping a backward recurrence over $k \in \{0, \dots, K\}$.
We compute all $T$ inner products simultaneously in $O(T \log T)$ time using \texttt{jax.scipy.signal.fftconvolve}. Computing $\langle\bfc, \bfc\rangle$ is straightforward and fast. 
Finally, our recurrences at step $t$ only rely on future steps $t+1$ and $t + b$.
Rather than allocating a full $T \times (K+1)$ table in device memory, we implement the backward sweep using a circular rolling buffer of size $b \times (K+1)$.
\end{itemize}

Together, these optimizations enable us to generate $\approx 10^6$ privacy loss samples per second (\cref{fig:throughput}), a $\approx50$-$100\times$ improvement over prior work, making Monte Carlo accounting practical for smaller $\delta$.

\section{Optimizing \bandmf with MC Accounting}
\label{sec:optimization}

Our goal is to optimize the \bandmf strategy matrix $\bfC \in \mathbb{R}^{T \times T}$ while accounting for amplification. The standard objective for \dpmf, as in \cite{denisov2022improved}, is the \emph{expected total squared error} on a workload matrix $\bfA$ (typically the ``prefix sum'' all-ones lower-triangular matrix):
\begin{equation}
    \text{TSE}(\bfC) \;=\; \|\bfA \bfC^{-1}\|_F^2 \cdot \sigma_{\varepsilon, \delta}^2(\bfC)
    \label{eq:joint_opt_objective}
\end{equation}
where $\|\bfA \bfC^{-1}\|_F$ computes the matrix mechanism's $\ell_2$-error for $\sigma = 1$, and $\sigma_{\varepsilon,\delta}(\bfC)$ is the minimum noise multiplier required for the resulting \bandmf mechanism to satisfy $(\varepsilon, \delta)$-DP.

For cyclic strategies privacy depends on $\bfC$ only through its maximum column norm and bandwidth $b$.
As a result, $\bfC$ can be optimized independently for different values of $b$, and $\sigma$ can be calculated efficiently using standard numerical methods. This decoupling does not apply for non-cyclic strategies analyzed under Monte Carlo accounting, where the privacy depends on the specific contents of $\bfC$.  
Hence, $b$ and $\bfC$ must be optimized jointly with $\sigma$.

Freely optimizing the $b$ parameters of banded Toeplitz $\bfC$ is non-trivial under Monte Carlo accounting.
For simplicity, we adopt the following single-parameter family of Toeplitz strategies proposed in \citet{mckennaScaling2024}: 
$\bfC[i,1] \;\propto\; r + \nicefrac{1}{i},$ for $i = 1, \dots, b$.
In previous work, this parameterization was observed to be near-optimal for the prefix workload.
While this near-optimality may not hold under Monte Carlo accounting, we adopt it as a convenient heuristic to reduce the search space.

The joint optimization over bandwidth and Toeplitz coefficients can be done as a discrete outer search over candidate bandwidths $b \ \le T/K$, and a continuous 1D search over $r \in (-\nicefrac{1}{b}, \nicefrac{1}{2})$ for each candidate $b$.
Because the objective is well-behaved and unimodal in $r$, we solve the 1D subproblem using standard scalar bounded optimization \citep{press1992numerical}.
Evaluating the objective at $(b, r)$ requires estimating $\sigma_{\varepsilon, \delta}(\bfC)$. Since only an estimate is required, we use a much smaller Monte Carlo sample budget ($N \approx \frac{1}{\delta}$) than what we need for \emph{formal calibration}. This gives an approximately optimal candidate $(b^*, r^*)$, for which we perform a final formal calibration of $\sigma$ prior to training.

\section{Experiments} \label{sec:experiments}

We evaluate seven batch selection strategies: four established baselines \deterministic (\texttt{No-Amp}), \cyclicpoisson, \ballsinbins, and \minseppoisson and our \cyclicra, \nestedra, and \minsepra.
To ensure a fair comparison, all methods are normalized to $K$ expected (or exact) participations per user over $T$ iterations.
We use a different banded Toeplitz strategy tailored for each batch selection strategy.
For \deterministic, we optimize RMSE for $b = \frac{T}{K}$, and heuristically apply this same Toeplitz strategy to \ballsinbins.
For \cyclicpoisson and \cyclicra, we search $b \in \{1, 2, \dots, \frac{T}{K}\}$ to find the optimal $(b^*, \bfC^*)$ minimizing TSE (noting that $b^*$ can differ between the two), and heuristically adopt these strategies for \minseppoisson and the standard variants of \nestedra and \minsepra, respectively. We use the approach described in \cref{sec:optimization} for the optimized variants of these strategies.

\paragraph{Accounting Details}
For methods evaluated via Privacy Loss Distribution (PLD) accounting (\deterministic, \cyclicpoisson, and the PLD bound for \cyclicra), we use the \texttt{dp\_accounting} library to calibrate $\sigma$ to the target $(\varepsilon, \delta)$ guarantee.
For methods evaluated via Monte Carlo accounting (\ballsinbins, \minseppoisson, \nestedra, and \minsepra), we formally calibrate the required noise multiplier $\sigma^*$ using the descending search procedure from Algorithm~5 of \citet{dong2025bminsep}.
For our new methods, we initialize $\sigma_0$ to the calibrated value under PLD accounting for \cyclicra, and sequentially steps downward by a factor of $1.01$ ($\sigma_{k+1} = \sigma_k / 1.01$) until reaching the first candidate noise multiplier that fails the verification certificate.
To certify $(\varepsilon, \delta)$-DP at our target $\delta = 10^{-5}$, we draw exactly $N = 10{,}745{,}968$ Monte Carlo samples per candidate $\sigma$. We run 3 independent random trials for each Monte Carlo calibration. 
This rigorous certification ensures that the calibrated noise levels (and resulting RMSE metrics) reported in our evaluation are provable upper bounds on privacy loss rather than just empirical point estimates. In figures and tables, Monte Carlo methods report the mean across the 3 independent calibration trials; the sample standard deviation is consistently $\le 0.01$ and thus statistical variance does not meaningfully affect any comparisons we make.

\paragraph{Batch Selection Strategy Configuration}

In \cref{tab:participation_normalization} we provide the settings of each batch selection strategy we use to normalize expected participations to $K$ to ensure fair comparisons.

\begin{table}[h]
\centering
\caption{Parameter settings for each batch selection strategy to achieve $K$ expected participations over $T$ iterations with minimum separation $b$.}
\label{tab:participation_normalization}
\begin{tabular}{@{}lccl@{}}
\toprule
Method & Participations & Native parameter & Setting for target $K$ \\
\midrule
\cyclicra & Exactly $K$ & $K$ directly & N/A \\
\nestedra & Exactly $K$ & $K$ directly & N/A \\
\minsepra & Exactly $K$ & $K$ directly & N/A \\
\ballsinbins & Exactly $K$ & Cycle length $L$ & $L = T/K$ \\
\cyclicpoisson & Expected K & Sampling prob.\ $p$ & $p = Kb/T$ \\
\minseppoisson & Expected K & Sampling prob.\ $p$ & $p = \frac{K}{T - K(b-1)}$ \\
\deterministic & Exactly $K$ & Cycle length $L$ & $L=T / K$ \\
\bottomrule
\end{tabular}
\end{table}

\paragraph{Choice of $\bfC$ for PLD-Based Methods}
For the \deterministic batch selection strategy, we optimize $\bfC^*$ numerically using standard techniques for $b=\frac{T}{K}$. For \cyclicpoisson and \cyclicra, do the same thing for $b \in \{1, 2, \dots, \frac{T}{K}\}$ and pick the best $(b^*, \bfC^*)$ pair (noting that $b^*$ can differ between the two).
\paragraph{Choice of $\bfC$ for Monte-Carlo-Based Methods}
We consider two settings of $\bfC$ for each batch selection strategy: \texttt{Std} and \texttt{Opt}. \texttt{Std} is a fixed, heuristically chosen $\bfC$ matrix optimized for the appropriate reference PLD-based approach.  For \ballsinbins, the reference strategy is \deterministic, for \minseppoisson the reference strategy is \cyclicpoisson, and for \nestedra and \minsepra the reference strategy is \cyclicra. Using a reference $\bfC$ is representative of a setting where optimization of Monte Carlo-calibrated $\bfC$ may be expensive or $\bfC$ might be restricted by system constraints. Second, we consider \texttt{Opt}, where ($b$, $\bfC$) is jointly optimized for the specific batch selection strategy as described in \cref{sec:optimization}. Since this approach is new, that means the numbers we report for \minseppoisson are actually better than the original papers (for \ballsinbins, our relative RMSE improvements for the most part match or exceed those reported by \cite{choquettechoo2025nearexact} for their optimization method). By looking at both \texttt{Std} and \texttt{Opt}, we isolate and quantify improvements from both improved batch selection strategy (\cref{sec:new_strategies}) and joint optimization (\cref{sec:optimization}.)

\paragraph{Evaluation}

We compare the batch selection strategies, while varying number of iterations ($T$), number of participations per example ($K$) and privacy budgets ($\varepsilon$). We center the experiments around the \emph{anchor} config $T=1024$, $K=4$, and $\varepsilon=4$, and vary each parameter one at a time. We chose this config because it represents a standard practical setting for large-scale DP training, where single-digit privacy budgets are common and high-epoch training is impractical. Our primary evaluation metric is the (normalized) RMSE, defined as $\text{RMSE} =\frac{1}{K} \sqrt{\nicefrac{\text{TSE}}{T}}$.

\subsection{RMSE Experiments}

\begin{table}[t]
\centering
\caption{Table of results showing the RMSE of each batch selection strategy normalized by \cyclicpoisson (values $< 1.00$ indicate improvement). The raw RMSE for \cyclicpoisson is shown in parenthesis for reference.}
\label{tab:rmse_sweep_eps4}
\resizebox{\textwidth}{!}{%
\begin{tabular}{@{}clccccccc@{}}
\toprule
 & & \multicolumn{4}{c}{\textbf{Prior Work}} & \multicolumn{3}{c}{\textbf{Proposed (Ours)}} \\
\cmidrule(lr){3-6} \cmidrule(lr){7-9}
 & & \multicolumn{2}{c}{PLD} & \multicolumn{2}{c}{Monte Carlo (Std / Opt)} & PLD & \multicolumn{2}{c}{Monte Carlo (Std / Opt)} \\
\cmidrule(lr){3-4} \cmidrule(lr){5-6} \cmidrule(lr){7-7} \cmidrule(lr){8-9}
 & Param & \makecell{\texttt{No-Amp}} & \makecell{\texttt{Cyclic} \\ \texttt{Poisson}} & \makecell{\texttt{Balls} \\ \texttt{in Bins}} & \makecell{\texttt{Min Sep} \\ \texttt{Poisson}} & \makecell{\texttt{Cyclic} \\ \texttt{RA}} & \makecell{\texttt{Nested} \\ \texttt{RA}} & \makecell{\texttt{Min Sep} \\ \texttt{RA}} \\
\midrule
 & 256   & 1.14 & 1.00 (1.32) & 0.85 / 0.78 & 0.91 / 0.75 & 0.88 & 0.76 / 0.73 & 0.75 / \textbf{0.68} \\
 & 512   & 1.18 & 1.00 (1.38) & 0.85 / 0.76 & 0.90 / 0.72 & 0.89 & 0.75 / 0.71 & 0.73 / \textbf{0.65} \\
\rowcolor{gray!10}  & 1024  & 1.21 & 1.00 (1.44) & 0.84 / 0.75 & 0.90 / 0.69 & 0.89 & 0.76 / 0.69 & 0.75 / \textbf{0.63} \\
\multirow{-4}{*}{\rotatebox{90}{Vary $T$}} & 2048  & 1.24 & 1.00 (1.50) & 0.83 / 0.74 & 0.90 / 0.68 & 0.89 & 0.74 / 0.66 & 0.74 / \textbf{0.61} \\
\midrule
 & 1     & 1.05 & 1.00 (3.14) & 0.55\footnotemark / \textbf{0.54} & 0.90 / \textbf{0.54} & 0.89 & 0.72 / \textbf{0.54} & 0.72 / \textbf{0.54} \\
 & 2     & 1.16 & 1.00 (2.04) & 0.69 / 0.65 & 0.90 / 0.61 & 0.88 & 0.69 / 0.61 & 0.68 / \textbf{0.54} \\
\rowcolor{gray!10}  & 4     & 1.21 & 1.00 (1.44) & 0.84 / 0.75 & 0.90 / 0.69 & 0.89 & 0.76 / 0.69 & 0.75 / \textbf{0.63} \\
 & 8     & 1.21 & 1.00 (1.11) & 0.96 / 0.86 & 0.93 / 0.80 & 0.92 & 0.85 / 0.79 & 0.84 / \textbf{0.77} \\
\multirow{-5}{*}{\rotatebox{90}{Vary $K$}} & 16    & 1.17 & 1.00 (0.94) & 1.03 / 0.97 & 0.98 / 0.89 & 0.95 & 0.92 / 0.90 & 0.92 / \textbf{0.87} \\
\midrule
 & 1.0   & 1.75 & 1.00 (3.43) & 1.22 / 0.90 & 0.98 / 0.87 & 0.93 & 0.90 / 0.85 & 0.90 / \textbf{0.84} \\
 & 2.0   & 1.48 & 1.00 (2.17) & 1.03 / 0.82 & 0.94 / 0.79 & 0.91 & 0.82 / 0.75 & 0.81 / \textbf{0.74} \\
\rowcolor{gray!10}  & 4.0   & 1.21 & 1.00 (1.44) & 0.84 / 0.75 & 0.90 / 0.69 & 0.89 & 0.76 / 0.69 & 0.75 / \textbf{0.63} \\
\multirow{-4}{*}{\rotatebox{90}{Vary $\varepsilon$}} & 8.0   & 1.00 & 1.00 (0.97) & 1.00 / 0.72 & 0.97 / 0.65 & 0.87 & 0.81 / 0.66 & 0.80 / \textbf{0.58} \\
\bottomrule
\end{tabular}}
\end{table}
\footnotetext{At ($K=1, b=T$), \ballsinbins, \nestedra, \minsepra are equivalent. \ballsinbins ($b=T$) outperforms the non-optimized alternatives because they use $b = 64$ (optimal for \cyclicra).}

Our main experimental results are shown in \cref{tab:rmse_sweep_eps4}.  Below we quantify improvements in terms of the \emph{relative RMSE} between two methods, i.e., the ratio between two RMSEs.

\begin{itemize}
\item Among PLD-based methods (which are cheaper to calibrate), \cyclicra consistently provides the best RMSE, beating \cyclicpoisson by between $5$-$14\%$, and about $11\%$ on average. 
\item  Prior to our work, the best batch selection strategy varied across settings. Our proposed \minsepra with optimized batch selection strategy consistently performs the best. The magnitude of improvement over \cyclicpoisson ranges from $13\%$ to $46\%$, with improvement growing with larger $T, \epsilon$ and smaller $K$. In all cases the second best strategy is either \nestedra or \minseppoisson (with an optimized strategy). 
\item When using the standard strategy (optimized for \cyclicra), the RMSE between \minsepra and \nestedra is consistently less than $3\%$ in all settings we considered.  When using the optimized strategies, the difference is up to $10\%$ in some settings.
\item The gain between \cyclicra and \nestedra / \minsepra using the same Toeplitz coefficients averages $12\%$.  When optimizing the Toeplitz coefficients, the relative RMSE compared to the standard Toeplitz coefficients averages $9.8\%$ for \nestedra and $14.2\%$ for \minsepra.
\end{itemize}

\begin{figure}[t]
  \centering
  \begin{subfigure}[b]{0.49\textwidth}
    \centering
    \includegraphics[width=\textwidth]{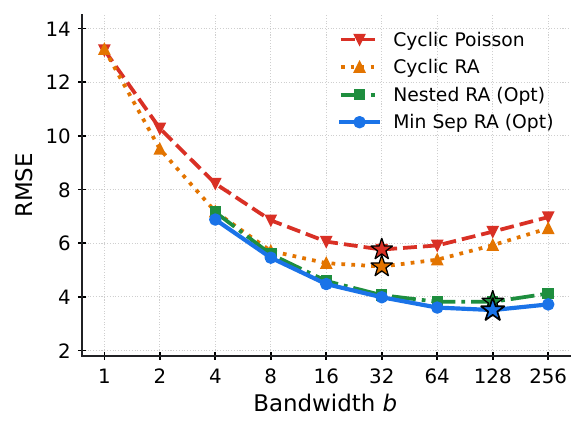}
    \vspace{-2em}
    \caption{}
    \label{fig:rmse_bandwidth}
  \end{subfigure}
  \begin{subfigure}[b]{0.49\textwidth}
    \centering
    \includegraphics[width=\textwidth]{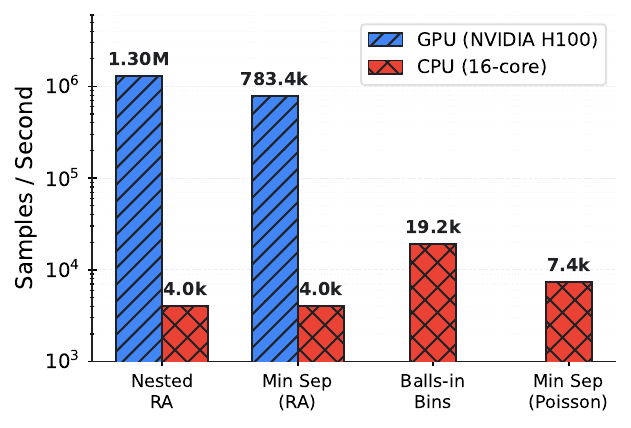}
    \vspace{-2em}
    \caption{}
    \label{fig:throughput}
  \end{subfigure}\hfill
  \caption{(a) RMSE as a function of number of bands $b$ and (b) Sampling throughput of Monte Carlo accounting on CPU and GPU. Both plots show results for $T=1024$, $K=4$, and $\varepsilon=4.0$.}
  \label{fig:throughput_and_rmse}
\end{figure}

In \cref{fig:rmse_bandwidth}, we show the RMSE as a function of the bandwidth $b$ under the anchor configuration. 

\begin{itemize}
\item While \cyclicpoisson and \cyclicra achieve their optimal RMSE at $b = 32$, \nestedra and \minsepra reach their optimum at a $4\times$ larger bandwidth of $b = 128$.
\item While \cyclicpoisson and \cyclicra require delicate bandwidth tuning to balance privacy amplification against correlated noise, for both \nestedra and \minsepra the maximum feasible bandwidth $T / K$ is a reasonable choice (although both still benefit from tuning bandwidth).
\end{itemize}

\subsection{Sampling Throughput}

In \cref{fig:throughput} we measure the throughput for Monte Carlo sampling for both our new batch selection strategies, and existing ones under the anchor config.  For existing methods, we use the implementation available in \texttt{JAX Privacy}, and for our methods, we use the one described in \cref{sec:implementation}. 

\begin{itemize}
    \item On CPU, our implementation generates about $4K$ samples per second, while the \texttt{JAX Privacy} implementations enjoy $19.2K$ and $7.4K$ samples per second, respectively. This difference is explained by the extra dimension $(K)$ needed in our dynamic programming state. At this throughput, generating $1$ million samples takes $\approx 4$ minutes on a single CPU.  Recalling that we need $\approx 10M$ samples for a single formal verification at $\delta=10^{-5}$, and an additional sweep over up to $\approx 50$ different $\sigma$ values, this is feasible but expensive and may require parallelization.
    \item The throughput increases substantially when moving to an H100 GPU, up to $1.3M$ samples/second for \nestedra and $783K$ samples/second for \minsepra. This makes it feasible to do both strategy optimization and formal calibration on a single machine in a reasonable amount of time. This is enabled by our highly-optimized implementation described in \cref{sec:implementation}.
\end{itemize}

\subsection{Gemma Finetuning Experiments}
\begin{figure}[t]
  \centering
  \includegraphics[width=0.80\textwidth]{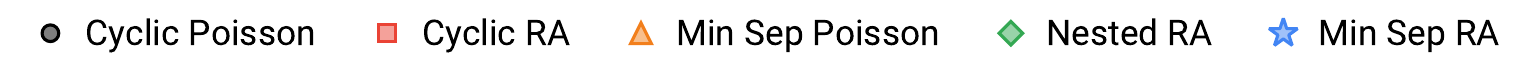}\\[0.5ex]
  \begin{subfigure}[b]{0.48\textwidth}
    \centering
    \includegraphics[width=\textwidth]{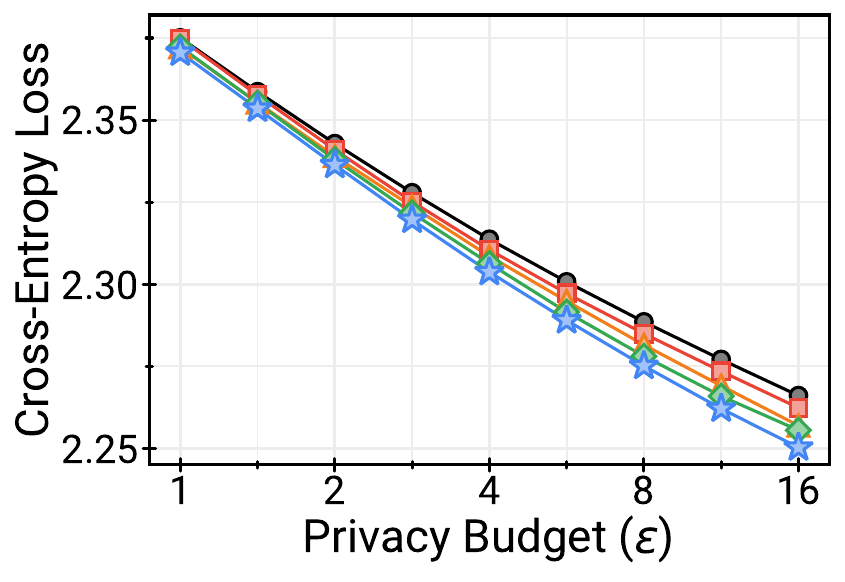}
    \caption{Loss vs.\ $\varepsilon$ (fixed $\bfC$)}
    \label{fig:pubmed_loss_vs_eps}
  \end{subfigure}\hfill
  \begin{subfigure}[b]{0.48\textwidth}
    \centering
    \includegraphics[width=\textwidth]{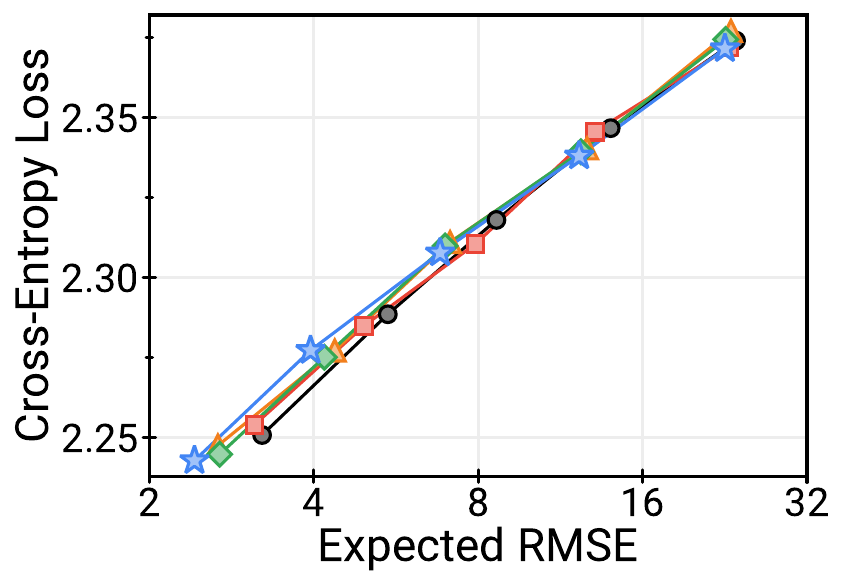}
    \caption{Loss vs.\ RMSE (optimized $\bfC^*$)}
    \label{fig:pubmed_loss_vs_rmse}
  \end{subfigure}
  \caption{PubMed fine-tuning ($T = 512$, $K = 8$, $\delta = 10^{-5}$) with denoised AdamW ($\beta = 0.85$). \textbf{(a)} Training cross-entropy loss vs.\ privacy budget $\varepsilon$ with fixed $\bfC$ ($b = 16$). \textbf{(b)} Training cross-entropy loss vs.\ expected RMSE across strategy-optimal $(b^*, \bfC^*)$.}
  \label{fig:pubmed_fixed_b16}
\end{figure}

To complement our analytic evaluations, we LoRA fine-tune (rank $16$) a 2B Gemma~4 model on PubMed ($N = 80{,}000$) for $T = 512$ iterations with $K = 8$, measuring mean loss over the final quarter of training with an EMA-denoised AdamW variant.\footnote{We compute AdamW's second-moment estimate from a bias-corrected exponential moving average ($\beta = 0.85$) of noisy gradients rather than raw single-step gradients, using weight decay $10^{-4}$, clipping norm $1.0$, and learning rate $\alpha=10^{-3}$ for fixed $\bfC$ and tuned over $\alpha \in \{\nicefrac{1}{4}, \nicefrac{1}{2}, 1, 2\} \times 10^{-3}$ when $\bfC$ varies.}
We vary the batch selection strategy (omitting \ballsinbins as a special case of \nestedra), evaluating $\epsilon \in [1, 16]$ at $\delta = 10^{-5}$ both with fixed $\bfC$ ($b = 16$, optimized without amplification) and with jointly optimized $(b^*, \bfC^*)$.

\begin{itemize}
\item With fixed $\bfC$ (\cref{fig:pubmed_loss_vs_eps}), improvements in $\sigma$ via tighter accounting directly translate to improved training loss, as expected. \minsepra consistently outperforms competing techniques, matching \cyclicpoisson's loss with $9\%$ to $39\%$ less privacy budget.
\item With strategy-optimal $(b^*, \bfC^*)$ (\cref{fig:pubmed_loss_vs_rmse}), improved RMSE translates directly into improved training loss according to a linear trend ($R^2 = 0.995$) shared by all 5 batch selection and noise correlation strategies. Using the EMA-denoised AdamW variant was important here; with standard Adam there is less alignment, as observed in \cite{mckennaScaling2024}.
\end{itemize}

\section{Discussion and Future Work}
\label{sec:discussion}

In this work, we unified two previously separate lines of research in private training: random allocation for \dpsgd and privacy amplification for banded matrix factorization. By generalizing the $b$-separation and random allocation families into a cohesive framework, our proposed strategies consistently improve RMSE across all settings evaluated, achieving gains over prior PLD-based approaches as large as $1.85\times$. Furthermore, our carefully-optimized implementation makes Monte Carlo accounting orders of magnitude more efficient, transforming it into a more practical tool.

Our results point to several promising directions for further improvement. First, formal calibration currently draws $\Omega(1/\delta)$ samples for each candidate noise multiplier $\sigma$. Narrowing down the initial candidate $\sigma_0$ could reduce the number of evaluated values from roughly $50$ to $2$ or $3$, providing a substantial speedup. Moreover, developing heuristics to directly predict the optimal $(b, \theta, \sigma_0)$ tuple from a given setting $(\varepsilon, \delta, T, K)$ would bypass our expensive $(b, r)$ grid search.
Second, our Toeplitz strategy optimization is somewhat coarse: we consider only powers of $2$ for bandwidth $b$ and restrict Toeplitz coefficients to a single-parameter family. Exploring other parameterizations of the strategy matrix could yield additional utility improvements.

\bibliographystyle{plainnat}
\bibliography{references}

\appendix

\crefalias{section}{appendix}




\section*{AI Use Statement}

In this work, we used generative AI tools for implementing methods and proposing the random rotation construction (\cref{sec:biased_batches}) to resolve the biased batch size issue. We have not used generative AI tools for experimental design, developing our core theoretical framework, or writing the primary proofs, and dataset-related and translation tasks are not applicable to this work. Additionally, we used generative AI tools for writing code, managing code health, generating figures and tables, and editing the manuscript to improve readability. We have reviewed all AI-assisted work, and all LLM-generated code and experimental outputs were verified and tested for correctness by the authors. We take responsibility for the final content of this work, including text, claims, or artifacts produced with the aid of generative AI.

\section{Monte Carlo Background}\label{sec:mc-background}

Below we restate Algorithm 5 of \cite{dong2025bminsep}, which allows for formal calibration of a parameter such as the noise multiplier with end-to-end DP guarantees using Monte Carlo accounting. In the context of \dpmf,  ``dominating pair'' means the pair $P, Q$ given by \cref{lem:dimensionreduction}. In our use-case, each $\calM_i$ is a run of \bandmf with a different noise multiplier, and we have $\calM_1, \calM_2, \ldots$ correspond to a decreasing sequence of noise multipliers.

\begin{algorithm}[h]
\caption{Monte Carlo Calibration}
\begin{algorithmic}[1]\label{alg:mc-calibration}
\REQUIRE Candidate mechanisms $\calM_1, \calM_2, \ldots \calM_k$. Target privacy parameters $(\varepsilon, \delta)$, ``base'' privacy parameter $\delta_{base} \in (0, \delta)$, input dataset $D$.
\STATE $N \leftarrow$ minimum number of samples such that any mechanism not satisfying $(\varepsilon, \delta)$-DP passes Monte Carlo verification of $(\varepsilon, \delta_{base})$-DP w.p. at most $\frac{\delta - \delta_{base}}{1 - \delta_{base}}$ (computable using open-source \texttt{jax\_privacy} utilities).
\STATE $i = 1$
\WHILE{$i \leq k$}
\STATE $S^+ \leftarrow N$ samples from PLD for dominating pair $(P_i, Q_i)$ for $\calM_i$
\STATE $S^- \leftarrow N$ samples from PLD for $(Q_i, P_i)$
\IF{Monte Carlo estimate from $S^+$ (resp. $S^-$) of $H_{e^\epsilon}(P_i, Q_i)$ (resp. $H_{e^\epsilon}(Q_i, P_i)$)$> \delta_{base}$}
\STATE \textbf{break}
\ENDIF
\STATE $i = i + 1$
\ENDWHILE
\IF{$i = 1$}
\RETURN $\bot$ (calibration failed)
\ENDIF
\RETURN $\calM_{i-1}(D)$
\end{algorithmic}
\end{algorithm}

\begin{theorem}[Theorem A.1 of \cite{dong2025bminsep}]\label{thm:mc-calibration}
\cref{alg:mc-calibration} (calibration via Monte Carlo accounting and the final run/release of the inner mechanism $\calM_{i-1}$) satisfies $(\epsilon, \delta)$-DP.
\end{theorem}

\subsection{Sample Complexity of Monte Carlo Accounting}

In \cref{fig:sample_complexity} we plot the number of samples per candidate $\sigma$ to do formal Monte Carlo calibration at different target $\delta$ values. For $\delta=10^{-5}$, we needed $10,745,967$ samples per candidate for formal calibration, which we rounded up to $11,000,000$ for simplicity in our experiments.

\begin{figure}[H]
\centering
\includegraphics[width=0.67\textwidth]{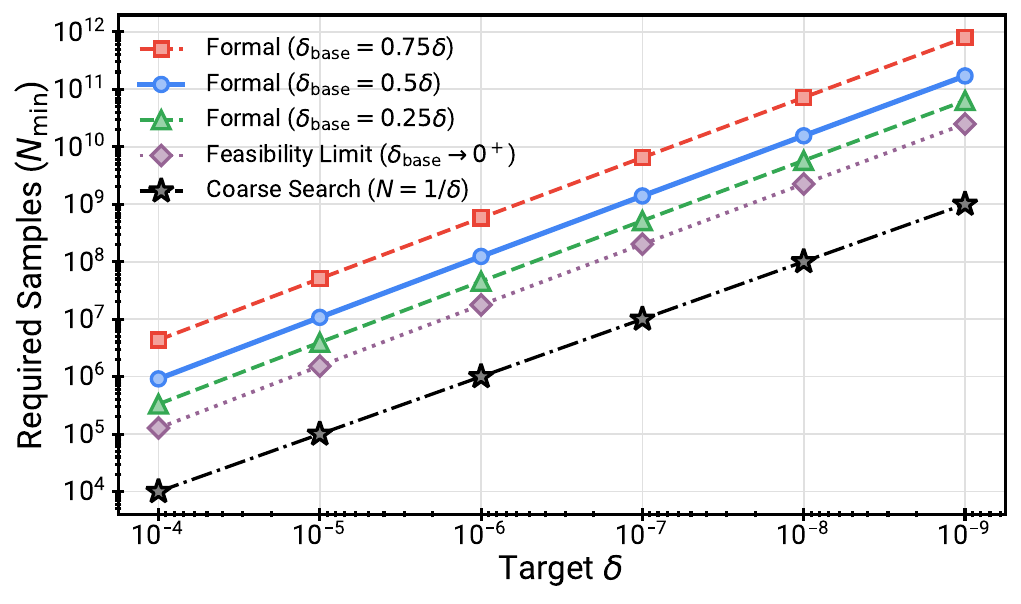}
\caption{\label{fig:sample_complexity} The number of samples needed for Monte Carlo accounting.  The blue line ($\delta_{base}=0.5\delta$) is the default setting for formal calibration recommended by prior work and used in this paper for formal calibration.}
\end{figure}

\section{Batch Selection Strategy Extensions}

\subsection{Generalized Batch Selection Strategies}\label{app:generalized_strategies}

\begin{figure}[t]
    \centering
    \includegraphics[width=\textwidth]{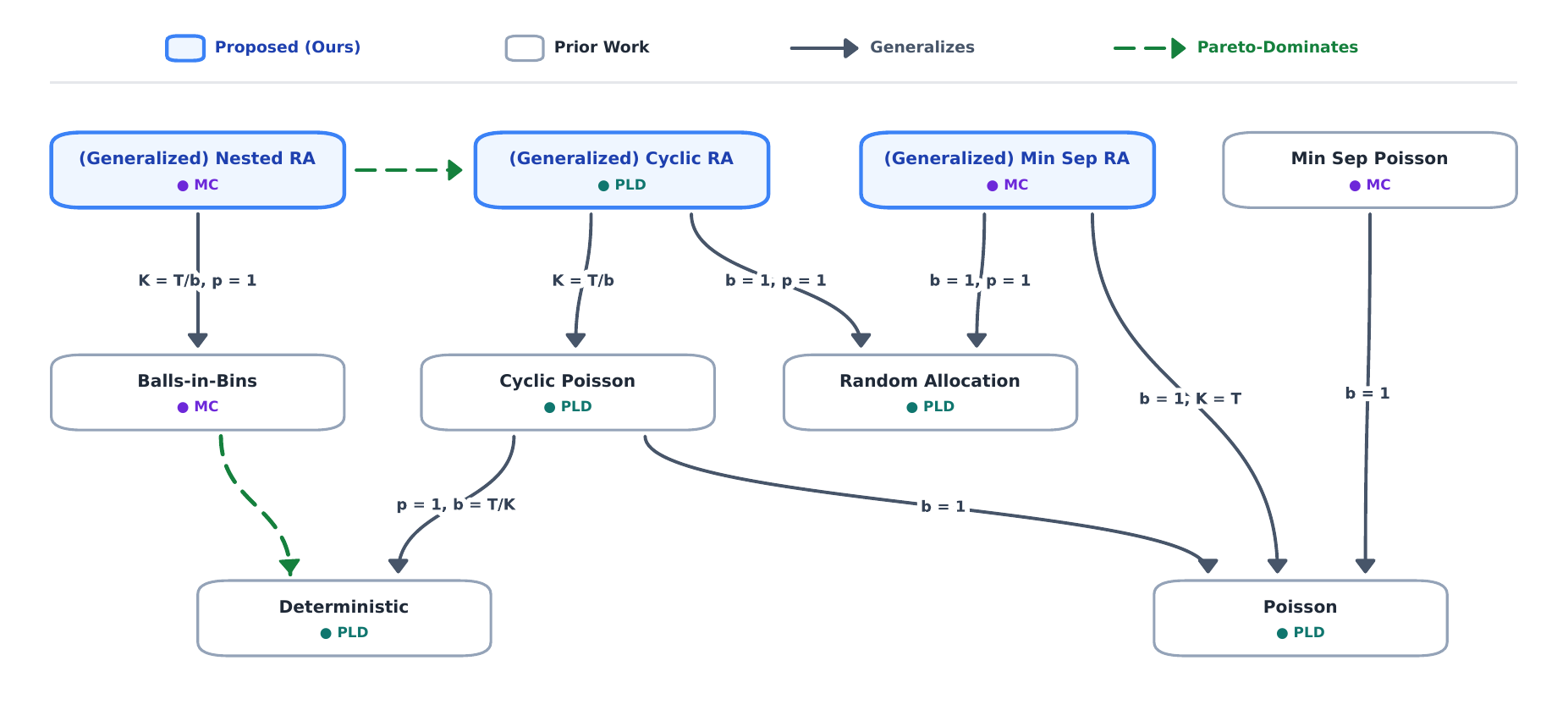}
    \caption{Relationships among our generalized batch selection strategies.  Each of our batch selection strategies are parameterized in terms of $(T, b, K, p)$, while the prior strategies are defined in terms of some subset of those parameters. Solid arrows indicate exact special cases under the labeled parameter settings; dashed arrows indicate Pareto dominance.}
    \label{fig:generalized_strategies}
\end{figure}

In \cref{sec:new_strategies}, we defined our random-allocation-based batch selection strategies with a fixed (integer) number of participations $K$ to simplify the presentation, and because it is what we recommend in practice. However, each strategy can be easily generalized by introducing a sampling probability $p \in (0, 1]$: after sampling a $b$-separated subset $S \subseteq [T]$ of size $K$ as before, each round in $S$ is retained independently with probability $p$. An example then participates $pK$ times in expectation while still satisfying $b$-separation (this also allows us to have $<1$ expected participations). Privacy analysis goes through with little change:

\begin{itemize}
    \item For \cyclicra, the same reduction as \cyclicpoisson still applies, composing $K$-out-of-$T/b$ random allocation with Poisson subsampling at rate $p$.
     \item For \nestedra and \minsepra, because retention decisions are independent across $t \in S$ and any subset of $S$ remains $b$-separated, replacing $\ell_t$ with $(1 - p) + p\,\ell_t$ in \cref{eq:nra-dp,eq:msra-dp} computes $P(\mathbf{y})/Q(\mathbf{y})$ while leaving the recurrences unchanged.
\end{itemize}

As illustrated in \cref{fig:generalized_strategies}, these generalizations clarify how our methods relate to prior batch selection strategies.  A solid black line indicates one batch selection strategy is a strict generalization of another, while a dashed green line is used to indicate that we \emph{conjecture} (without proof) that one batch selection strategy is never worse than another in terms of RMSE.

Enumerating the equivalences, we have:
\begin{itemize}
    \item \texttt{RA}$(T, K)$ $\equiv$ \cyclicra$(T, K, b=1, p=1)$ $\equiv$ \minsepra($T, K, b=1, p=1$)
    \item \cyclicpoisson$(T, b, p)$ = \cyclicra$(T, b, K=\nicefrac{T}{b}, p)$
    \item \poisson$(T,p)$ $\equiv$ \minsepra($T, K=T, b=1, p)$
\end{itemize}

Because our proposed batch selection strategies are strict generalizations of the prior ones, we immediately inherit any theoretical guarantees on RMSE that pertain to the existing instantiations of \bandmf. For example, \citet{kalinin2024banded} give an asymptotic bound on RMSE of a specific instantiation of \bandmf with a fixed (non-optimized) choice of $\bfC$ under the \deterministic batch selection strategy.  \citet{mckennaScaling2024} further argued that optimizing over the space of $\bfC$ matrices could be no worse than using a fixed $\bfC$ in terms of RMSE, and hence directly inherits those theoretical guarantees, while enjoying strictly better constant factors in practice.  A similar argument applies here.


We also observe that \ballsinbins is identical to \deterministic except that each example's modulus $i \in [b]$ is randomized, and similarly \nestedra is identical to \cyclicra with the same randomization. By Blackwell's informativeness theorem \citep{blackwell1953equivalent}, the privacy guarantees of \ballsinbins (resp. \nestedra) strictly dominate those of \deterministic (resp. \cyclicra). We further believe (and showed empirically) that \minsepra is preferable to \nestedra in practical settings of interest. However, \minsepra cannot cleanly be expressed as \nestedra with additional randomness, which means we cannot immediately apply Blackwell's informativeness theorem to make any formal statement about Pareto dominance between these batch selection strategies.

\subsection{Biased Batch Sizes for \minsepra} \label{sec:biased_batches}

As mentioned in \cref{sec:new_strategies}, \minsepra has one critical drawback that the other batch selection strategies don't have: the expected batch size varies across iterations, with iteration numbers near the endpoints enjoying larger batch sizes in expectation. The severity of this issue depends on the values of $T$, $K$, and $b$.  At $b=1$, there is no issue (standard random allocation), while as $b$ approaches $\frac{T}{K}$, the issue becomes more prominent. \cref{fig:biased_batches} demonstrates this issue for $T=1024$ and $K=4$. While the bias is negligible for $b \leq 8$ and small for $b=32$, it is considerable at $b=128$ and becomes completely degenerate at $b=256$ (not shown). It's natural to ask which setting we might find ourselves in: \cref{tab:repro_min_sep_ra} shows the optimal value of $b$ in our setting is $128$, so this is potentially a serious concern in practice.

\begin{figure}[h]
\centering
\includegraphics[width=0.5 \textwidth]{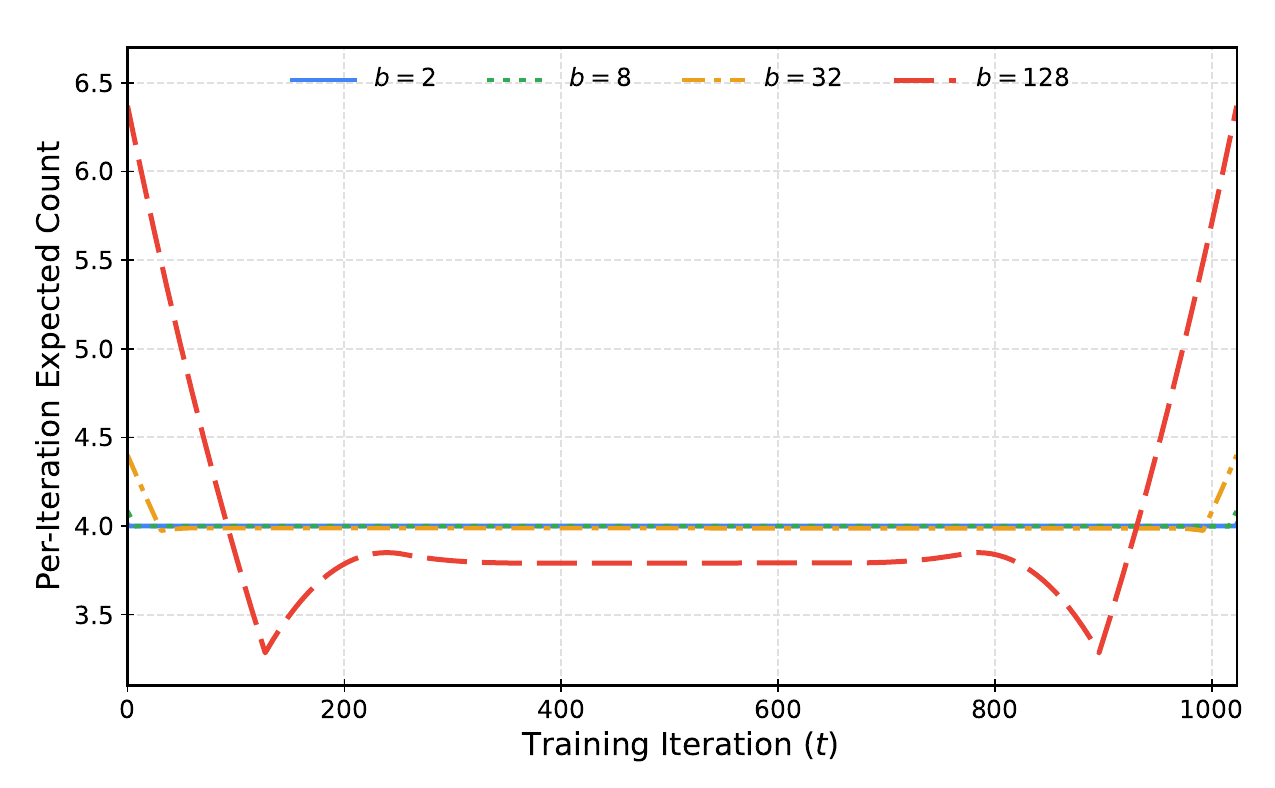}
\caption{Per-Iteration expected batch size for \minsepra with $T=1024$ and $K=4$.}
\label{fig:biased_batches}
\end{figure}

A simple modification to the sampling algorithm addresses this issue fairly naturally. Standard \minsepra samples uniformly from the set of indices $s_0, \dots, s_{K-1}$ satisfying $s_i + b \leq s_{i+1}$ for $i=0, \dots, K-2$. Adding a simple ``wraparound'' constraint $s_{K-1} + b - T \leq s_0$ ensures if a record participates in an early iteration, they cannot also participate in a late iteration, which intuitively addresses the bias issue.

To sample efficiently from this modified set, we can first sample $K$ positive gaps $(g_1, \dots, g_K)$ summing to $T$ such that every gap satisfies the separation constraint $g_j \ge b$. This can be accomplished via stars and bars by distributing the slack $T - K \cdot b$ across $K$ non-negative integer parts and adding $b$ to each.
Accumulating these gaps yields a valid $b$-separated participation subset, to which we apply an independent uniform shift $R \sim \mathrm{Uniform}([T])$ modulo $T$, giving strictly uniform expected batch sizes across iterations. That is, our subset is $\{R, R+g_1, R+g_1+g_2, \ldots R+\sum_{i=1}^{K-1} g_i\}$, except for each of these values which exceeds $T$, we subtract $T$. 

\begin{theorem}
    The above method is equivalent to sampling a uniformly random subset from those that satisfy the $b$-separation constraint (including the wraparound constraint). Furthermore, there are $\frac{T}{k} \cdot \binom{T-Kb+K-1}{K-1}$ such subsets. 
\end{theorem}
\begin{proof}
For any $b$-separated subset $t_1 < t_2 < \ldots t_k$, there are exactly $k$ values of $(R, g_1, g_2, \ldots, g_K)$ that produce that subset. Namely, the $i$th such value is $R = t_i, g_1 = t_{i+1 \pmod K} - t_i, g_2 = t_{i+2 \pmod K} - t_{i+1 \pmod K}, \ldots $, where we add $T$ to any negative $g_i$.

This proves a uniform distribution over $(R, g_1, g_2, \ldots, g_K)$ is equivalent to a uniform distribution over $b$-separated subsets, as desired. It also proves the number of such subsets: There are $\binom{T-Kb+K-1}{K-1}$ choices of $g_1, \ldots, g_k$ using stars-and-bars, and $T$ choices of $R$. So there are $T \cdot \binom{T-Kb+K-1}{K-1}$ values of $(R, g_1, g_2, \ldots, g_K)$, and $1/K$ times as many $b$-separated subsets.

\end{proof}

While this change to the sampling procedure is minor, formal privacy accounting must account for these changes. We describe a new scoring dynamic program to correctly account for this random rotation.

Define $\ell_t$ as we defined it for Eqs. \eqref{eq:nra-dp}, \eqref{eq:msra-dp}. Let $\calS$ be all subsets of $[T]$ in the support of our batch selection strategy. Recalling that $|\calS| = \frac{T}{K} \cdot \binom{T-Kb+K-1}{K-1}$, our goal is to compute:

\[\frac{P(\bfy)}{Q(\bfy)} = \frac{\sum_{S \in \calS} \prod_{t \in S} \ell_t}{\frac{T}{k} \cdot \binom{T-Kb+K-1}{K-1}}.\]

It suffices to give a method to compute the numerator. We decompose the numerator as:

\[\sum_{S \in \calS} \prod_{t \in S} \ell_t = \sum_{i=1}^{b-1}\sum_{S \in \calS: \min(S) =i} \prod_{t \in S} \ell_t + \sum_{S \in \calS: \min(S) \geq b} \prod_{t \in S} \ell_t.\]

We solve for all $b$ sums over sets $S$ via dynamic program, using the following observation: for any $i \in [b-1]$, if $\min(S) = i$, then the remaining $K-1$ elements of $S$ must be in $[i + b, T - b + i]$ and are all separated by at least $b$, but without the wraparound constraint. Similarly, if $\min(S) \geq b$, then we have $K$ elements of $S$ in $[b, T]$ separated by at least $b$ without the wraparound constraint. Since we have reduced to a setting without the wraparound constraint, we can reduce to an analysis similar to that of Eq. \eqref{eq:msra-dp}. In particular, let $\text{SUM}(k, lo, up)$ denote the sum over sets $S$ of $k$ elements in $[lo, up]$ which are at least $b$ apart of $\prod_{t \in S} \ell_t$. Then we have base cases $\text{SUM}(0, lo, up) = 1$ (only the empty set, for which the sum is 1) and $\text{SUM}(k, lo, up) = 0$ if $k > 0$ and $up - lo < (k-1)b$ (no such set exists, so the sum is empty). We have the recurrence:

\[\text{SUM}(k, lo, up) = \text{SUM}(k, lo+1, up) + \ell_{lo} \cdot \text{SUM}(k - 1, lo + b, up).\]

The proof of correctness for this recurrence is similar to that of \cref{thm:msra-dp} (since both recurrences are solving the same problem, only here we are solving for a sum and not a ratio), so we omit it here. Now we have:

\begin{align*}
\sum_{S \in \calS} \prod_{t \in S} \ell_t =& \sum_{i=1}^{b-1}\sum_{S \in \calS: \min(S) =i} \prod_{t \in S} \ell_t + \sum_{S \in \calS: \min(S) \geq b} \prod_{t \in S} \ell_t\\ =& \sum_{i=1}^{b-1} \ell_i \cdot \text{SUM}(K-1, i+b, T-b+i) + \text{SUM}(K, b, T).
\end{align*}

Each recurrence takes $O(1)$ time to evaluation, we evaluate $\text{SUM}$ at $O(bKT)$ triplets of inputs (the third input only takes values in $[T-b+1, T]$, and the final combination of values of $\text{SUM}$ takes $O(b)$ time, so the overall runtime for computing $P/Q$ is $O(bKT)$. Note this is a $O(b)$ multiplicative overhead compared to \cref{thm:msra-dp}; we leave it to future work to determine if this overhead can be removed.

\section{Supporting Experiments}

\subsection{Optimized Toeplitz Parameters}

In \cref{tab:repro_nested_ra,tab:repro_min_sep_ra} we provide the precise parameters needed to reproduce any cell of our main optimized \nestedra and \minsepra methods respectively. Using this reference table, readers can easily reproduce our reported results by running our formal calibration binary with the provided inputs, and do not need to independently figure out what the optimal settings we found were.

\begin{table}[t]
\centering
\caption{Formal Calibration Parameters and Reproduction Details for \nestedra ($\delta = 10^{-5}$, $N_{\min} = 10{,}745{,}968$). The anchor configuration ($T = 1024, K = 4, \varepsilon = 4.0$) is highlighted in gray. Toeplitz coefficients $\mathbf{C}[i, 1] \propto \frac{1}{i} + r^*$ are optimized over $(b, r)$ pairs and certified under Algorithm~5.}
\label{tab:repro_nested_ra}
\begin{tabular}{@{}clcccccccc@{}}
\toprule
 & $T$ & $K$ & $\varepsilon$ & $b^*$ & $r^*$ & $\|AC^{-1}\|_F / \sqrt{T}$ & $\sigma^*$ & RMSE & RMSE/$K$ \\
\midrule
 & 256  & 4  & 4.0 & 32 & $+0.0397$ & 3.435 & 1.119 & 3.84 & 0.96 \\
 & 512  & 4  & 4.0 & 64 & $+0.0375$ & 3.758 & 1.037 & 3.90 & 0.97 \\
\rowcolor{gray!10}  & 1024 & 4  & 4.0 & 128 & $+0.0350$ & 4.026 & 0.981 & 3.95 & 0.99 \\
\multirow{-4}{*}{\rotatebox{90}{Vary $T$}} & 2048 & 4  & 4.0 & 128 & $+0.0513$ & 4.817 & 0.828 & 3.99 & 1.00 \\
\midrule
 & 1024 & 1  & 4.0 & 512 & $+0.0254$ & 3.358 & 0.505 & 1.70 & 1.70 \\
 & 1024 & 2  & 4.0 & 256 & $+0.0350$ & 3.444 & 0.720 & 2.48 & 1.24 \\
\rowcolor{gray!10}  & 1024 & 4  & 4.0 & 128 & $+0.0350$ & 4.026 & 0.981 & 3.95 & 0.99 \\
 & 1024 & 8  & 4.0 & 32 & $+0.1190$ & 5.349 & 1.321 & 7.06 & 0.88 \\
\multirow{-5}{*}{\rotatebox{90}{Vary $K$}} & 1024 & 16 & 4.0 & 16 & $+0.2464$ & 6.602 & 2.045 & 13.50 & 0.84 \\
\midrule
 & 1024 & 4  & 1.0 & 32 & $+0.0136$ & 6.823 & 1.703 & 11.62 & 2.90 \\
 & 1024 & 4  & 2.0 & 64 & $+0.0438$ & 4.825 & 1.361 & 6.57 & 1.64 \\
\rowcolor{gray!10}  & 1024 & 4  & 4.0 & 128 & $+0.0350$ & 4.026 & 0.981 & 3.95 & 0.99 \\
\multirow{-4}{*}{\rotatebox{90}{Vary $\varepsilon$}} & 1024 & 4  & 8.0 & 128 & $+0.0756$ & 3.603 & 0.710 & 2.56 & 0.64 \\
\bottomrule
\end{tabular}
\end{table}

\begin{table}[t]
\centering
\caption{Formal Calibration Parameters and Reproduction Details for \minsepra ($\delta = 10^{-5}$, $N_{\min} = 10{,}745{,}968$). The anchor configuration ($T = 1024, K = 4, \varepsilon = 4.0$) is highlighted in gray. Toeplitz coefficients $\mathbf{C}[i, 1] \propto \frac{1}{i} + r^*$ are optimized over $(b, r)$ pairs and certified under Algorithm~5.}
\label{tab:repro_min_sep_ra}
\begin{tabular}{@{}clcccccccc@{}}
\toprule
 & $T$ & $K$ & $\varepsilon$ & $b^*$ & $r^*$ & $\|AC^{-1}\|_F / \sqrt{T}$ & $\sigma^*$ & RMSE & RMSE/$K$ \\
\midrule
 & 256  & 4  & 4.0 & 32 & $+0.0184$ & 3.634 & 0.991 & 3.60 & 0.90 \\
 & 512  & 4  & 4.0 & 64 & $+0.0196$ & 4.076 & 0.886 & 3.61 & 0.90 \\
\rowcolor{gray!10}  & 1024 & 4  & 4.0 & 128 & $+0.0211$ & 4.426 & 0.821 & 3.63 & 0.91 \\
\multirow{-4}{*}{\rotatebox{90}{Vary $T$}} & 2048 & 4  & 4.0 & 256 & $+0.0163$ & 4.964 & 0.741 & 3.68 & 0.92 \\
\midrule
 & 1024 & 1  & 4.0 & 512 & $+0.0269$ & 3.343 & 0.510 & 1.71 & 1.71 \\
 & 1024 & 2  & 4.0 & 256 & $+0.0226$ & 3.695 & 0.599 & 2.21 & 1.11 \\
\rowcolor{gray!10}  & 1024 & 4  & 4.0 & 128 & $+0.0211$ & 4.426 & 0.821 & 3.63 & 0.91 \\
 & 1024 & 8  & 4.0 & 32 & $+0.0709$ & 5.757 & 1.194 & 6.87 & 0.86 \\
\multirow{-5}{*}{\rotatebox{90}{Vary $K$}} & 1024 & 16 & 4.0 & 32 & $+0.0121$ & 6.866 & 1.908 & 13.10 & 0.82 \\
\midrule
 & 1024 & 4  & 1.0 & 32 & $+0.0232$ & 6.569 & 1.761 & 11.57 & 2.89 \\
 & 1024 & 4  & 2.0 & 64 & $+0.0119$ & 5.747 & 1.113 & 6.40 & 1.60 \\
\rowcolor{gray!10}  & 1024 & 4  & 4.0 & 128 & $+0.0211$ & 4.426 & 0.821 & 3.63 & 0.91 \\
\multirow{-4}{*}{\rotatebox{90}{Vary $\varepsilon$}} & 1024 & 4  & 8.0 & 128 & $+0.0339$ & 4.050 & 0.560 & 2.27 & 0.57 \\
\bottomrule
\end{tabular}
\end{table}

\subsection{Epsilon-Delta Trade-off Curves}

While the main paper focuses primarily on RMSE of \bandmf mechanisms calibrated to specific $(\varepsilon, \delta)$ targets, it is also informative to look at the $(\varepsilon, \delta)$ trade-off curve that \bandmf has when instantiated with different Toeplitz coefficients and batch selection strategies. \cref{fig:hockey_stick} shows that when using the Toeptliz strategy optimized for \cyclicra, both \nestedra and \minsepra provide significant benefits over the PLD-based approaches, and are quite close in behavior with \minsepra providing slightly better trade-offs.  When the Toeplitz coefficients are specifically optimized for \minsepra, the gap widens substantially.

\begin{figure}[H]
\includegraphics[width=0.5\textwidth]{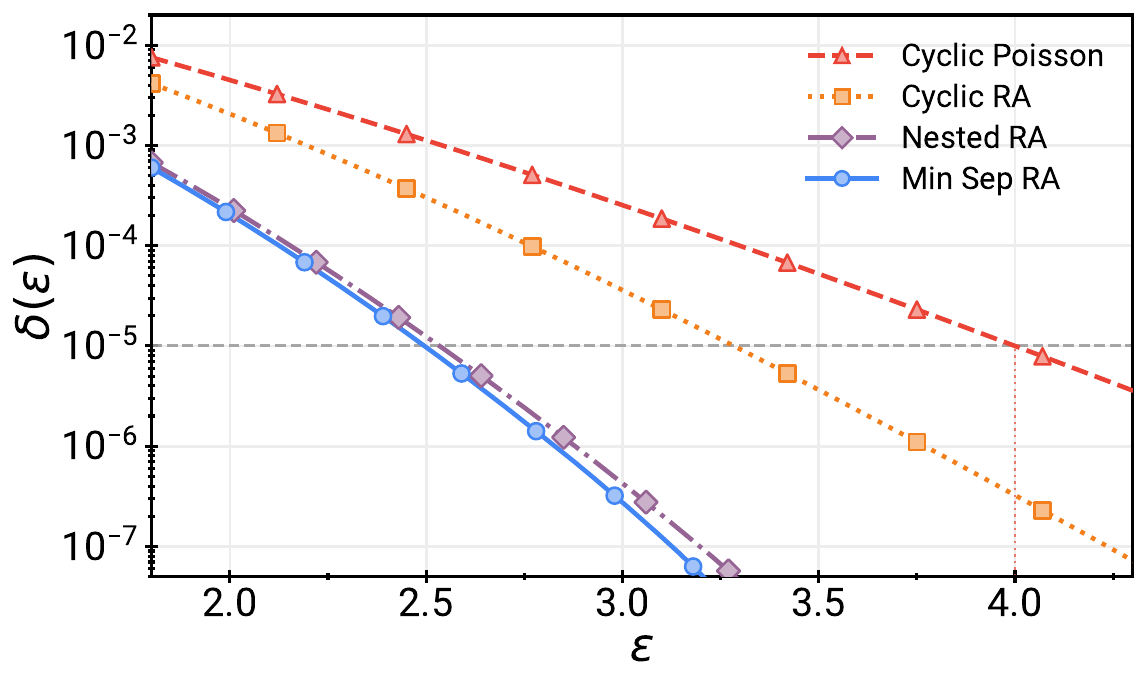}
\includegraphics[width=0.5\textwidth]{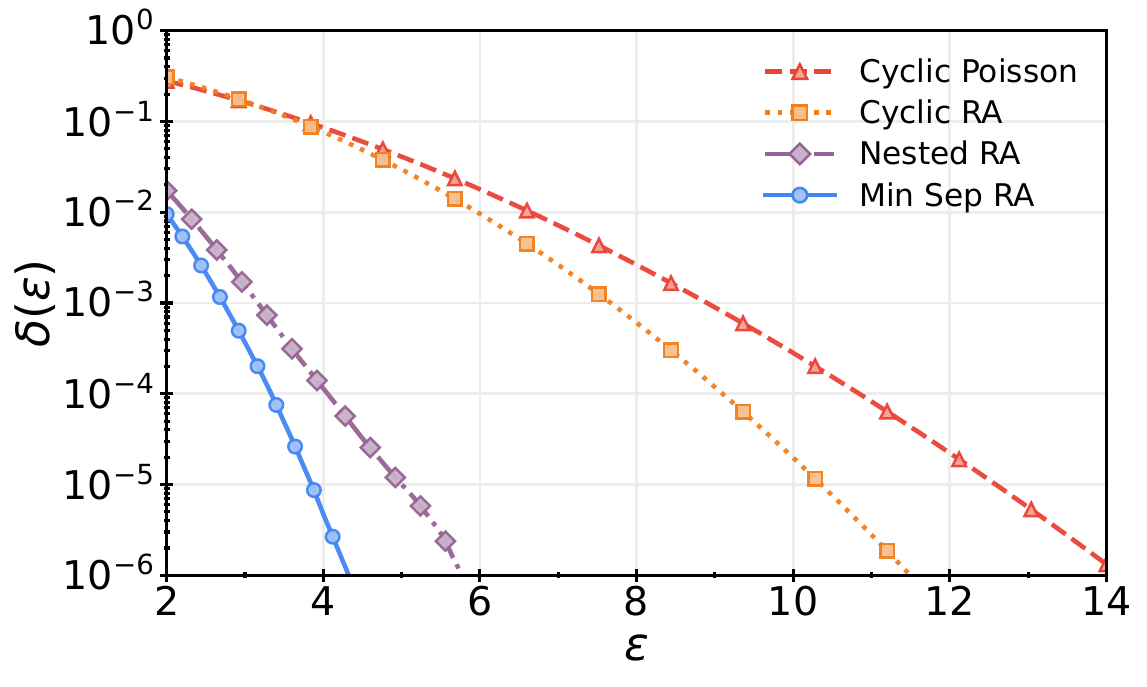}
\caption{\label{fig:hockey_stick} $(\varepsilon, \delta)$-trade-off curves for different batch selection strategies using Toeplitz coefficients optimized for \cyclicra (left) and \minsepra (right) at a target privacy level of $(4, 10^{-5})$. }
\end{figure}

\subsection{Monte Carlo optimized  Toeplitz Coefficients}

When using \cyclicpoisson or \cyclicra, the Toeplitz coefficients can be optimized independently of the accounting (since PLD accounting only requires knowing the L2 norm bound on the columns of $\bfC$). We can observe this directly in \cref{fig:optimal_coefs}, where these two lines overlap \footnote{This is not always the case, as in some settings the optimal number of bands may differ between these batch selection strategies.}

When using Monte Carlo accounting, we showed that we can benefit substantially from optimizing Toeplitz coefficients and accounting for the impact on privacy amplification. \cref{fig:optimal_coefs} shows how the optimal coefficients vary under different batch selection strategies.  We see that under \nestedra and \minsepra, the Toeplitz coefficients drop off more sharply. This drop off occurs because amplification for \nestedra and \minsepra improves as coefficients drop off faster, whereas the coefficients themselves do not affect amplification cyclic methods as long as the column norm is fixed. 

To see why coefficients dropping off helps amplification, intuitively, privacy amplification helps because it increases the adversary's uncertainty about what signal an example contributes to $\bfC \bfx$. When pairs of columns of $\bfC$ have high dot products, this reduces the uncertainty. As an extreme example, if $\bfC$ were the all-ones matrix, then $\bfC \bfx$ would only depend on the number of participations (which we keep fixed) and then privacy amplification wouldn't help at all. More broadly, if $\bfC[:, i]$ and $\bfC[:, j]$ have high dot product, then randomizing whether an example participates in round $i$ or $j$ does not add much more uncertainty for the adversary. When the coefficients drop off, the dot products between columns decreases and thus privacy amplification is increased. \minsepra has more entropy in $\bfx$ than \nestedra does, so it benefits more from this uncertainty, hence why its coefficients drop off even faster.

\begin{figure}[H]
\centering
\includegraphics[width=0.5\textwidth]{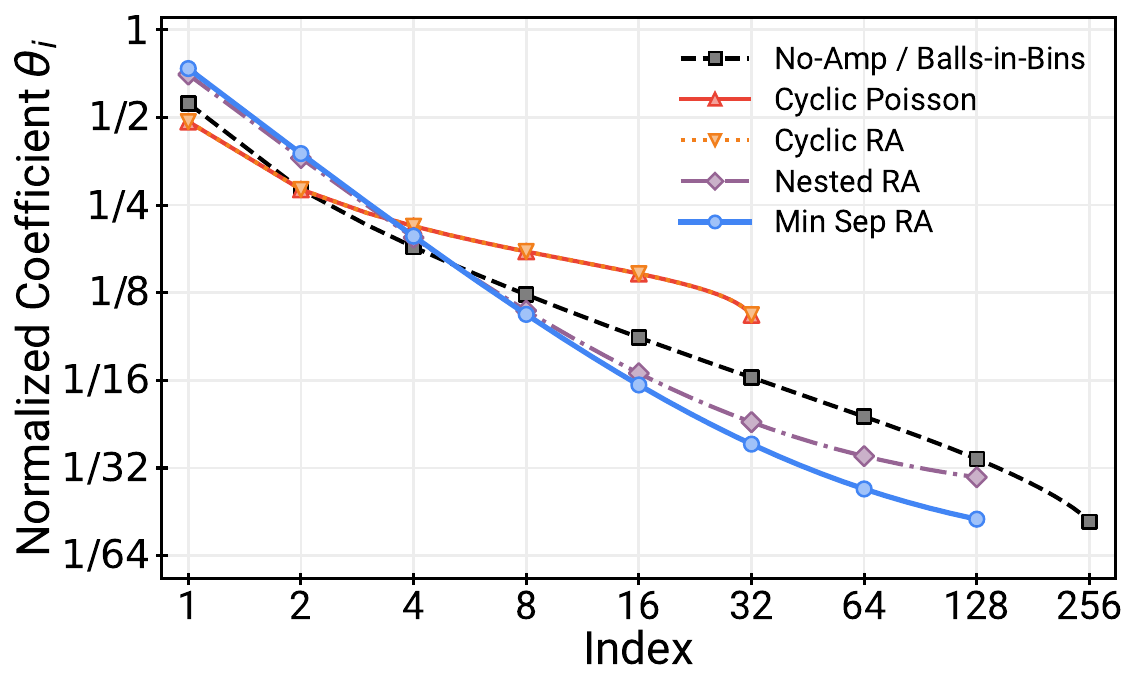}
\caption{Toeplitz coefficients optimized for different batch selection strategies when $T=1024$, $K=4$, and $\varepsilon=4$.}
\label{fig:optimal_coefs}
\end{figure}

\newpage

\end{document}